\documentclass{amsart}

\usepackage{amsmath,amssymb,amsfonts,latexsym} 
\usepackage{color}
\usepackage{graphicx} 
\usepackage{subfigure}

\usepackage{tikz}
\usetikzlibrary{arrows,positioning,shapes.geometric}
\usetikzlibrary{matrix,chains,decorations.pathreplacing}

\usepackage{xcolor}
\usepackage{xfrac}

\newtheorem{theorem}{Theorem}[section]

\theoremstyle{definition}
\newtheorem{definition}[theorem]{Definition}

\newtheorem{assumption}{Assumption}
\newtheorem{example}{Example}

\begin{document}

 %% and the full title of the paper in { }.
\title[Wave-shape oscillatory model] %Use the shortened version of the full title
      {Wave-shape oscillatory model for nonstationary periodic time series analysis}

\author{Yu-Ting Lin}
\address{Department of Anesthesiology, Taipei Veteran General Hospital, Taipei, Taiwan; School of Medicine, National Yang Ming Chiao Tung University, Taipei, Taiwan}

\author{John Malik}
\address{Department of Mathematics, Duke University, Durham, NC, USA}

\author{Hau-Tieng Wu}
\address{Department of Mathematics, Duke University, Durham, NC, USA; Department of Statistical Science, Duke University, Durham, NC, USA; Mathematics Division, National Center for Theoretical Sciences, Taipei, Taiwan}
\email{hauwu@math.duke.edu}

\maketitle
\begin{abstract}
The oscillations observed in many time series, particularly in biomedicine, exhibit morphological variations over time.    
These morphological variations are caused by intrinsic or extrinsic changes to the state of the generating system, henceforth referred to as dynamics. 
To model these time series (including and specifically pathophysiological ones) and estimate the underlying dynamics, we provide a novel {\em wave-shape oscillatory model}.
In this model, time-dependent variations in cycle shape occur along a manifold called the {\em wave-shape manifold}. To estimate the wave-shape manifold associated with an oscillatory time series, study the dynamics, and visualize the time-dependent changes along the wave-shape manifold, we {propose a novel algorithm coined  {\em Dynamic Diffusion map} (DDmap) by applying} the well-established diffusion maps (DM) algorithm to the set of all observed oscillations.  We provide a theoretical guarantee on the dynamical information recovered by the {DDmap} algorithm under the proposed model.
Applying the proposed model and algorithm to arterial blood pressure (ABP) signals recorded during general anesthesia leads to the extraction of nociception information. 
Applying the wave-shape oscillatory model and the {DDmap} algorithm to cardiac cycles in the electrocardiogram (ECG) leads to ectopy detection and a new ECG-derived respiratory signal, even when the subject has atrial fibrillation.
\end{abstract}

%The title of your section 1
\section{Introduction}

Oscillatory time series are ubiquitous in various scientific fields, such as medicine, epidemiology, {cosmology, geology,} and economics. If the period of oscillation is fixed, the time series is said to exhibit {\em seasonality}; if the period is not fixed, the term {\em cyclicity} is used.
{The period of an oscillatory physiological time series is rarely fixed.} {The ubiquitous cardiovascular waveform data in the modern hospital environment are typical examples. See} Figure \ref{Figure:pulse} for an example. {Besides the period of oscillation, a} visual inspection of {this} time series reveals several {additional} quantities that change {with} time: the amplitude, {the trend, and the shape of each oscillatory cycle -- a distinct quantity that we refer to as the {\em wave-shape}.}
{In biomedicine, these changes are the direct {result} of modulations in the state of the human system that generate the time series.
In this sense, detecting and quantifying modulations in the time series is akin to observing the dynamics of that human system.}

{For seasonal time series, existing approaches to} quantifying the available dynamical information include the seasonal auto-regressive integrated moving average (SARIMA) model \cite{Brockwell_Davis:2002}, the trigonometric Box-Cox transform, the trend and seasonal components algorithm (TBATS) \cite{DeLivera_Alysha_Hyndman_Snyder:2011}, {and} the complex de-modulation approach \cite{hasan1983complex}, among others.
However, these methods are inherently limited when the time series' period {and amplitude of oscillation, not to mention the wave-shape, are non-constant}. Modern time-frequency analysis tools \cite{flandrin1998time}, {particularly nonlinear-type time-frequency analysis tools \cite{wu2020current}} like the synchrosqueezing transform (SST), have been considered for {non-seasonal signals} \cite{Chen_Cheng_Wu:2014}; however, as we will carefully discuss in {the supplementary material, both the model behind the original SST and its generalization \cite{lin2018wave} fail when modulations in frequency, amplitude, or wave-shape do not occur gradually (a phenomenon} commonly seen when the underlying human system assumes {\em pathophysiological status}). Since well-behaved human systems are rarely of clinical interest, this challenge necessitates a different model for analyzing oscillatory physiological time series and, in particular, a tool for extracting the dynamics encoded by the time-varying morphology of their cycles.

\begin{figure}[h!]
\centering
\includegraphics[width=\textwidth]{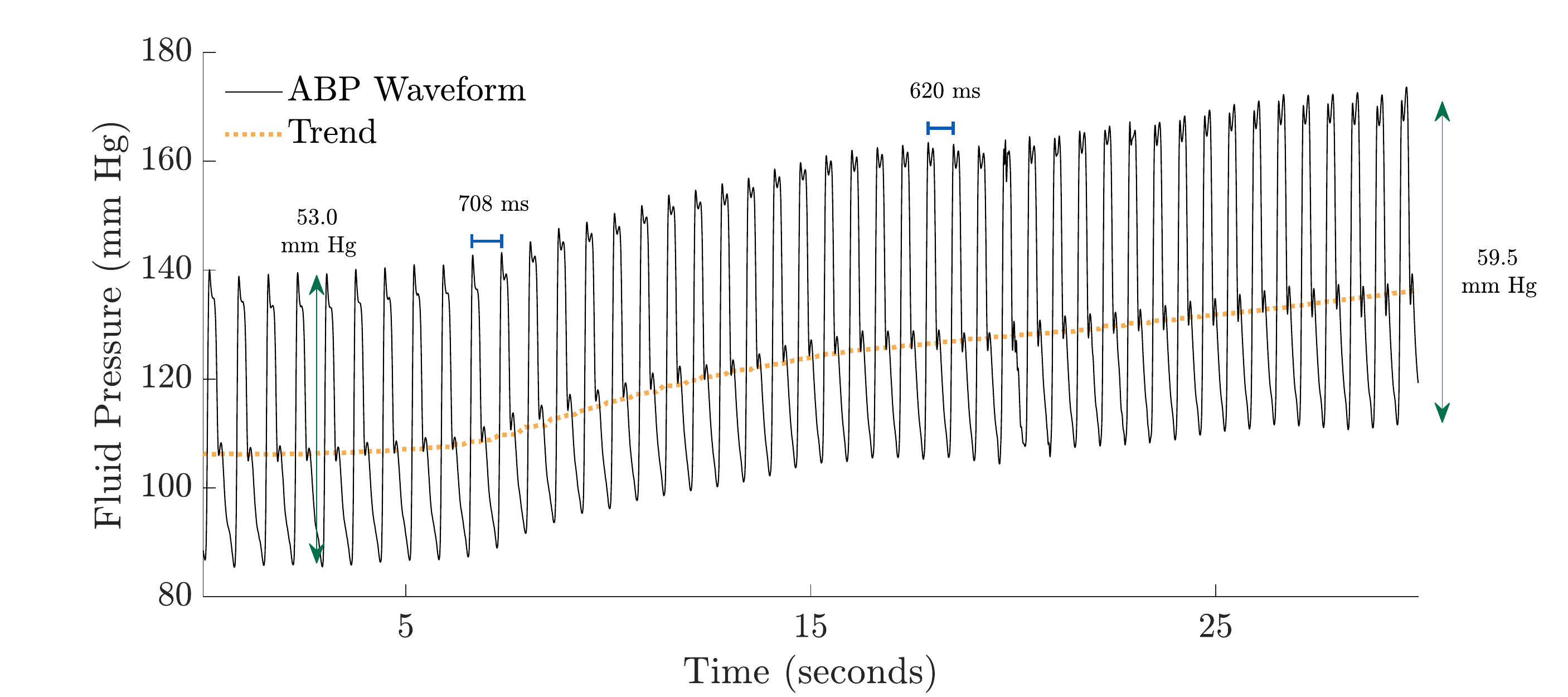}\\
\vspace{0.1in}\hspace{0.22in}\includegraphics[width=.44\textwidth]{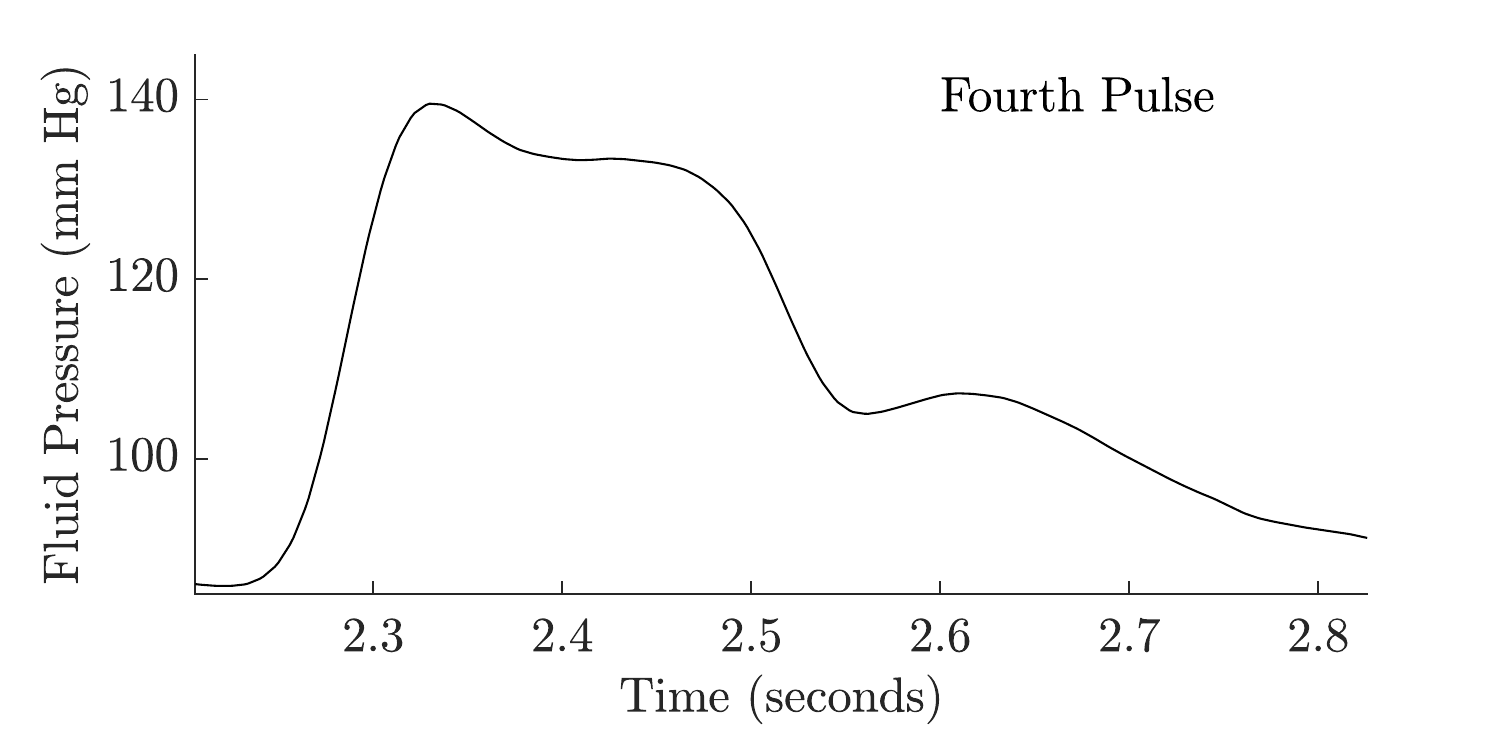}
\includegraphics[width=.44\textwidth]{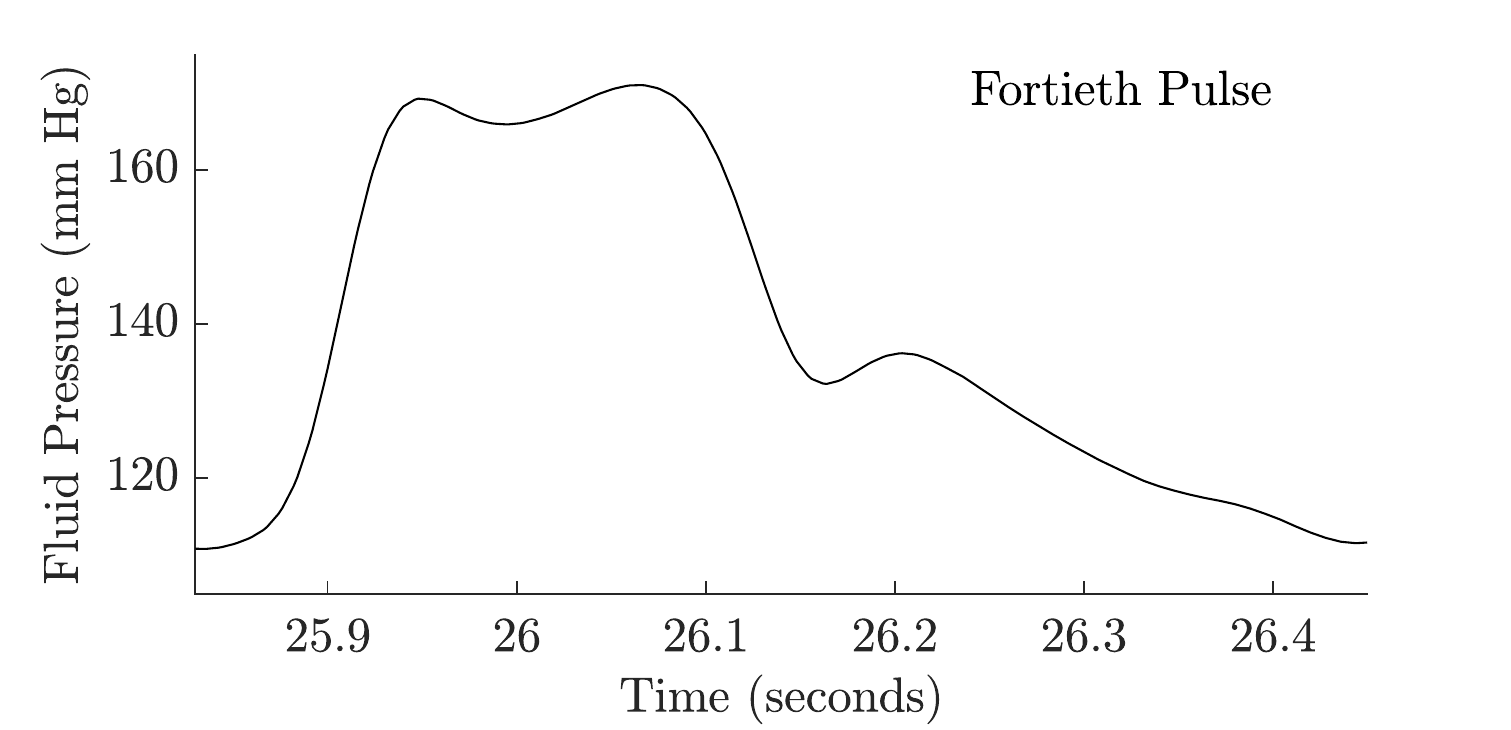}
\caption{An arterial blood pressure (ABP) signal obtained from {an anesthetized} patient undergoing the endotracheal intubation procedure. This oscillatory physiological time series changes in amplitude and frequency over time in response to the noxious stimulation associated with intubation. In addition, the shape of its non-sinusoidal oscillatory cycles (pulses) is also time-varying. This distinct type of modulation has strong clinical importance, and we refer to it as {\em wave-shape modulation}.}\label{Figure:pulse} 
\end{figure}

\subsection{Evidence from the clinic for the necessity of a new time series model}\label{Section:Motivation Physiological}

Biomedical time series are routinely used in the clinic to monitor a variety of human systems. %, 
Our motivation for a {new %the wave-shape oscillatory
time series} model comes in particular from the existing clinical approaches to cardiovascular waveform analysis. In this field, the oscillatory cycles from cardiovascular time series such as the ABP and ECG signals are processed in order to infer physical properties of the cardiovascular system. 
For example, vascular wall tension, blood volume, and {cardiac contractility} all affect the shape of the oscillatory cycles in the ABP signal and can therefore be inferred \cite{Nicholas_ORourke_Vlachopoulos:2011}. Unifying the existing approaches to pulse waveform analysis is the search for canonical features of the oscillatory cycle that empirically measure physical properties of the cardiovascular system.   
In the clinic, the relative and absolute positions of canonical landmarks on the pulse wave-shape (such as the well-known dicrotic notch) serve as indicators of cardiovascular status \cite{avolio2009role}.
Examples of these indicators include augmentation pressure \cite{Weber_Auer_ORourke_Kvas:2004} and the augmentation index \cite{o2005updated}.
These indicators are widely used because of their simplicity and empirical consistency.  On one hand, {the features of one patient's wave-shape are compared with the features of another patient's wave-shape in order to facilitate} diagnosis;
this approach has been proposed as a guide for the medical treatment of high blood pressure \cite{avolio2009role}. On the other hand, {the features of one patient's wave-shape are compared with the features of his or her successive wave-shapes in order to} summarize that patient's progress during surgery or treatment.
For example, commercial monitoring instruments such as the FloTrac\texttrademark system (Edwards Life Sciences, Inc., Irvine, CA, USA) bring physicians a continuous and relatively non-invasive assessment of the hemodynamic system by estimating cardiac output, mean arterial pressure, stroke volume, stroke volume variation, and systemic vascular resistance from the pulse waveform. This information facilitates the clinical management of hemodynamic instability \cite{van2017value,slagt2014systematic,teboul2016less}.

The first limitation consistent across these existing approaches is that they rely on {\em a priori} knowledge about the physiological significance of contrived pulse wave-shape features, limiting their applicability across mediums and rendering them useless for analyzing nascent or alternative signal types.
Second, the relationship between cardiovascular state and the pulse wave-shape may be more delicate or comprehensive than is implied by the {\em ad hoc} constructed feature.
Third, the impact of a significant change in cardiovascular state may be undetectable by the human eye or inconsistent across subjects, in which case the traditional approach to detecting that change would fail. 
Fourth, features that depend on the existence and detection of canonical landmarks struggle when encountering pathological or noisy waveforms. Specifically, if the waveform is pathological, the more delicate landmarks may be missing or may not have the same physiological significance. Moreover, the analysis of long-term signals would require the development of numerous sensitive and precise landmark detection algorithms, the lack of which would result in compounded inaccuracies.
Due to these limitations in the field of pulse waveform analysis, there is evidently a need for {an alternative} statistical model that makes precise the manifestation of dynamics in cardiovascular time series exhibiting wave-shape modulation and an accompanying algorithm that can succeed in estimating those dynamics.
{While we motivate the necessities of a new model using the biomedical signals, similar limitations might hold for times series in other fields}.
%While we envision the wave-shape oscillatory model's applicability to many of these types of time series {outside biomedical field}

\subsection{Contributions}

{First, we introduce the {\em wave-shape oscillatory model}, a novel approach to describing an oscillatory time series exhibiting wave-shape modulation (see Section \ref{sec:waveshapemanifold}). 
The model generalizes both traditional approaches to cardiovascular waveform analysis and modern mathematical approaches to oscillatory time series analysis. We make precise the standing assumption that changes in cycle shape are reflections of the dynamics of the underlying system. We introduce the notion of a wave-shape manifold: a structure hosting the set of observed cycles whose dimensionality is intuitively restricted by the number of independent factors giving rise to wave-shape modulation. 
In Section \ref{sec:dm}, we suggest analyzing a time series adhering to the wave-shape oscillatory model by applying the diffusion maps (DM) algorithm to its set of oscillatory cycles, and we show that by doing so, one is guaranteed to recover the dynamics hidden in that time series. ({The theory is postponed to} Section \ref{Section:Theory} of the supplementary material.) {We envision the proposed wave-shape oscillatory model's applicability to many time series of this type outside the biomedical field.}
We study two clinical databases to demonstrate that the proposed model allows us to generate a compact representation of the dynamical information hidden in an oscillatory physiological time series that may be pathological (see Section \ref{Sect:Afib} and Section \ref{sec:application2} of the supplementary material).
The main article is closed with a discussion in Section \ref{sec:discussion}.
In Section \ref{Section:Review} of the supplementary material, we provide a mathematical discussion of the wave-shape oscillatory model's relationship to previous attempts to modeling oscillatory time series with time-varying wave-shape. }

\section{The wave-shape oscillatory model}\label{sec:waveshapemanifold}

In this section, we provide a novel {model for oscillatory time series featuring wave-shape modulation that is suitable for analyzing pathological subjects. {The first consideration for our model is that an oscillatory time series can be viewed as a sequence of cycles. If $f \colon \mathbb{R} \rightarrow \mathbb{R}$ is an oscillatory time series, we assume that there is a set of functions $\mathcal{M} \subset L^2(\mathbb{R})$ so that each cycle of $f$ (after shifting the cycle from its time of occurrence to the origin) is a member of $\mathcal{M}$. For simplicity, we will assume that the support of each cycle is contained in the interval $[-\sfrac{1}{2}, \sfrac{1}{2}]$. An analysis is not yet feasible if we do not impose any conditions on $\mathcal{M}.$ Therefore, we take the following physiological fact into consideration: the internal physical and chemical conditions of an organism are well-constrained, even when under pathophysiological status. This condition leads us to assume that $\mathcal{M}$ is a low-dimensional structure embedded in $\mathcal{L}^2(\mathbb{R})$, and, to enable the analysis that will follow, we additionally assume that $\mathcal{M}$ is a smooth and compact Riemannian manifold. In Definitions~\ref{Def:waveshapemodel} and~\ref{Def:Waveshapemanifold}, we make these notions precise. 
\begin{definition}\label{Def:waveshapemodel}
A time series $f \colon \mathbb{R} \rightarrow \mathbb{R}$ adheres to the {\em  wave-shape oscillatory model} with wave-shape manifold $\mathcal{M} \subset L^2(\mathbb{R})$ if 
\begin{gather}
f(t) = \sum_{j \geq 1} s_j(t - t_j) +\varepsilon(t),
\end{gather}
for all $t \in \mathbb{R}$, where $t_j \in \mathbb{R}$ is the time at which the $j$-th cycle occurs, $s_j \in \mathcal{M}$ is the $j$-th cycle, and $\varepsilon \colon \mathbb{R} \rightarrow \mathbb{R}$ is stationary {white} random noise with mean zero and finite variance.
\end{definition}

\begin{definition}\label{Def:Waveshapemanifold} 
Let $\mathcal{H}$ be the subspace of $\mathcal{C}^{2}(\mathbb{R})$ consisting of functions whose support is a subset of $[-\sfrac{1}{2}, \sfrac{1}{2}]$. A {\em wave-shape manifold} is a smooth, compact, and low-dimensional Riemannian manifold embedded in $\mathcal{H}$. 
\end{definition}}

Note that $\mathcal H$ is a pre-Hilbert subspace {of} $L^2(\mathbb{R})$ {and that the smoothness imparted to each oscillatory cycle is necessary for drawing connections with previous mathematical approaches to modeling oscillatory time series}. 
We do not impose {any additional} conditions on the appearance of $\mathcal{M}$; in general, $\mathcal{M}$ may be disconnected {and have a boundary}. 
{Our intention with the wave-shape manifold is to} acknowledge the complicated {nonlinearity of the} structure underlying the time-dependent variations {in wave-shape}.  See Figure \ref{Figure:waveshape manifold simulation} for a visual example of a contrived wave-shape manifold. {Note that we used principal component analysis since we want to faithfully reflect the nonlinearity of the manifold via a linear projection.}
The structure of the wave-shape manifold is obviously guided by its classical landmark structure. However, it is also guided by other morphological features that we may not be able to quantify easily via landmarks. 
{Finally,} we {reiterate} that the wave-shape {oscillatory} model {has} not {arisen} from thin air{, but is inspired by traditional approaches to cardiovascular waveform analysis and generalizes recent mathematical approaches to modeling oscillatory time series.} The intimate relationship between the proposed wave-shape oscillatory model and the model behind {nonlinear-type time-frequency analysis} \cite{lin2018wave} is elaborated upon in the supplementary material {(see Section \ref{Section:Review})}.
\begin{figure}[h!]
\centering
\includegraphics[width=1\textwidth, trim=3cm 2cm 3cm 1cm, clip]{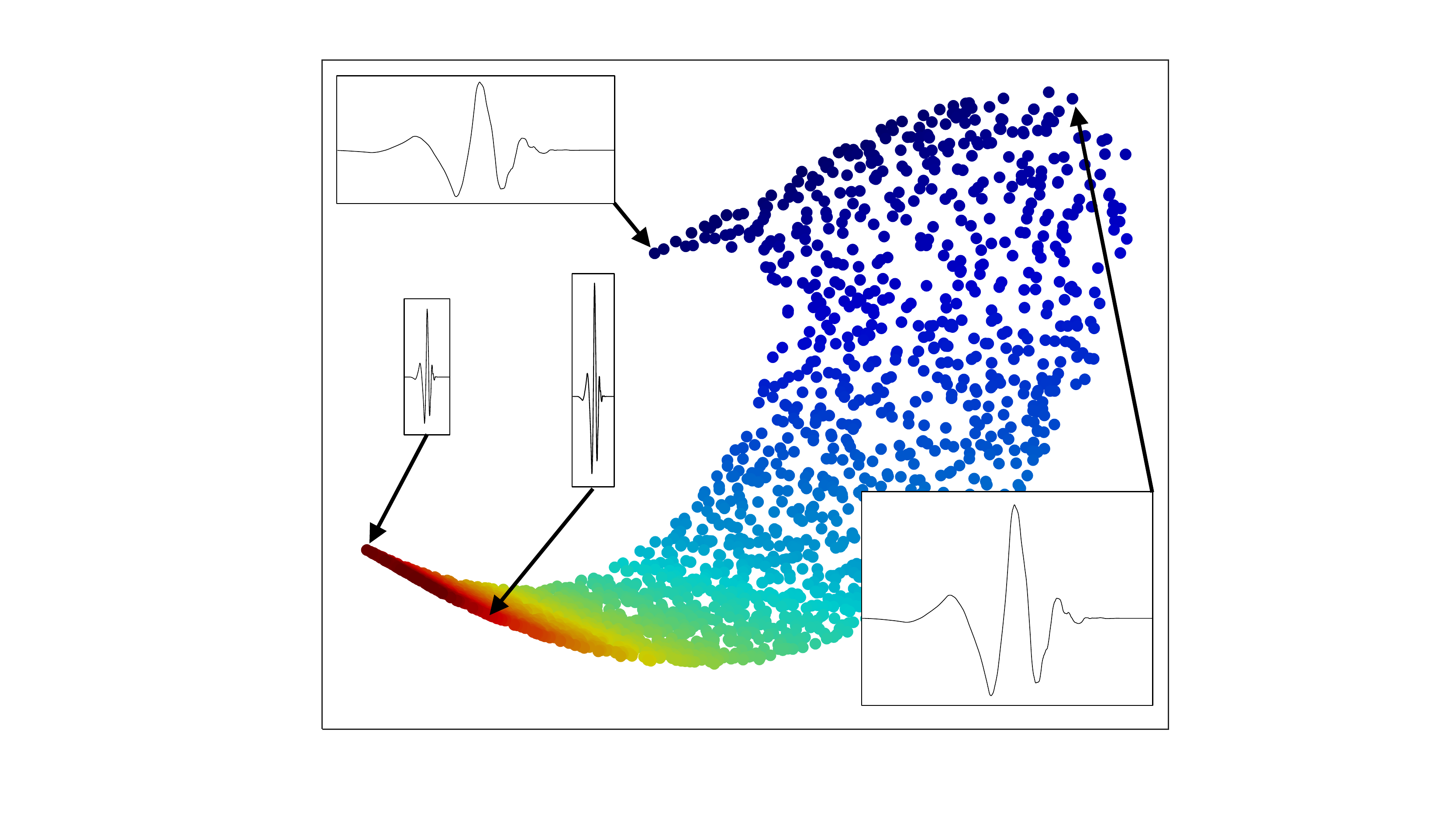}
\caption{An illustration of a {contrived} wave-shape manifold.  We build a collection of oscillatory cycles by dilating and scaling a template. 
We show a three dimensional embedding of this collection obtained by principle component analysis. {Points in the embedding are colored according to their associated dilation factor.}} 
\label{Figure:waveshape manifold simulation}
\end{figure}

\subsection{The underlying dynamics}

{The generation of the sequences $\{s_j\}_{j \geq 1}$ and $\{t_j\}_{j\geq 1}$ is a discrete dynamical process}
\begin{equation}\label{Introduction:dynamics}
(s_j, t_j)  = T\left((s_{j-1},t_{j-1}),(s_{j-2},t_{j-2}),...\right) \in \mathcal{M} \times \mathbb{R}
\end{equation}
{wherein} succeeding {oscillatory} cycles and their locations may depend on the location and morphology of any number of preceding cycles.
{The process $T$ will be different for each} physiological signal, and it is not possible to exhaustively discuss all cases here. {However, in the case of the ECG signal, some progress has been made.}
{Models describing the generation of $\{t_j\}_{j\geq 1}$ (independently of $\{s_j\}_{j \geq 1}$) have been developed to aid in heart rate variability (HRV) analysis; these models include} the long-range correlation model \cite{peng1995quantification}, the {history-dependent point-process model} \cite{Barbieri2005}, and others \cite[Chapter 4]{Clifford:2006:AMT:1213221}. However, to the best of our knowledge, the generation of {the oscillatory patterns described by} $\{s_j\}_{j \geq 1}$ has not been explicitly modeled. 

{We emphasize that Model \eqref{Introduction:dynamics} is suitable for modeling pathological recordings, which is less considered in the literature. For example,} {an ECG signal featuring premature ventricular contractions (a common cardiac arrhythmia) like those shown in Figures~\ref{Afib1} and \ref{PVCExample} has a wave-shape manifold with at least two connected components, and abnormal triggering points in the ventricles cause the sequence of observed oscillations to jump from one component of the manifold to the other. Modeling this ``exterior force'' is yet another challenge.}

\begin{figure}[h!]
\centering
\includegraphics[width=.8\textwidth, trim=7cm 6cm 7cm 5cm, clip]{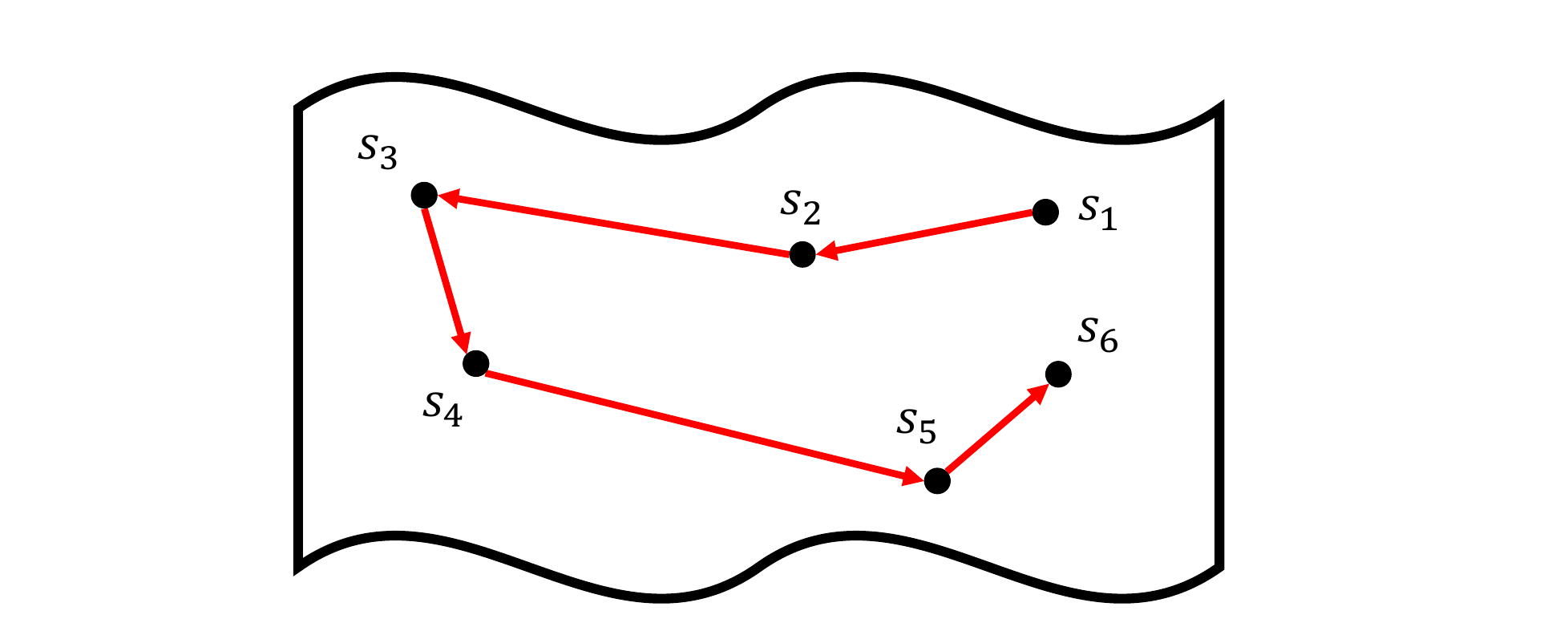}
\caption{An illustration of the process generating the parameters of $f$, a time series adhering to the wave-shape oscillatory model. On the surface of the wave-shape manifold $\mathcal{M}$, we show the sequence of wave-shape functions $\{s_j\}_{j\geq 1}$ observed in the time series.} 
\label{Figure:wave-shape oscillatory model}
\end{figure}

\if false
Clearly, the wave-shape oscillatory model \eqref{Introduction:ANHM1} differs from the phenomenological model \eqref{Introduction:ANHM0} in how the oscillation is modeled. The main benefit of \eqref{Introduction:ANHM1} is twofold. First, the dynamics encoded in the oscillatory morphology can be fully modeled and captured. Second, even if the oscillatory pattern dramatically changes from one cycle to the next, it can be captured by \eqref{Introduction:ANHM1}.

With the wave-shape oscillatory model, we expect to determine the dynamics encoded by the time-varying morphology of oscillatory cycles. The main challenges are the nonlinearity of the wave-shape manifold and the possible noise that deviates the observed oscillatory pattern from the wave-shape manifold. In the next section, we handle these challenges and achieve our goal using the diffusion geometry based algorithm, the DM.
\fi

\section{Estimating the underlying dynamics}\label{sec:dm}

{{The ultimate goal is obtaining $T$ from the recorded signal $f$. However, one challenge to obtaining $T$ is {the dependence between}} $\{s_j\}_{j\geq 1}$ and $\{t_j\}_{j \geq 1}${. For example,} in the case of the ECG signal, there is a nonlinear relationship between the QT interval (i.e., the support of $s_j$) and the {R-to-R} interval (i.e., $t_j - t_{j-1}$) \cite{malik2001problems}. {In this paper, we provide an algorithm to depict the profile of $T$.}
We begin by providing a brief problem statement. Suppose a time series $f \colon \mathbb{R} \rightarrow \mathbb{R}$ adheres to the wave-shape oscillatory model (see Definition~\ref{Def:waveshapemodel} and Definition~\ref{Def:Waveshapemanifold}). We assume that the times $\{t_j\}_{j\geq 1} \subset \mathbb{R}$ demarcating $f$'s oscillatory cycles have been correctly determined. 
%What if not?  What does it mean to be correct? 
We assume that we can build a set $\mathcal{X} = \{x_j\}_{j\geq 1} \subset L^2(\mathbb{R})$, where
\begin{gather}
x_j(t) = \left\{\begin{array}{cl}
s_j(t) + \varepsilon(t + t_j) & \text{if } t \in \left[-\sfrac{1}{2}, \sfrac{1}{2}\right]\\
0 & \text{otherwise} \end{array}\right.\label{noisywave}
\end{gather}
for all $j \geq 1$ (see Figure~\ref{Afib1}), where the set $\{s_j\}_{j\geq 1}$ has been dynamically sampled from $\mathcal{M}$ according to some never-zero and sufficiently smooth probability density function, {and $\varepsilon$ is a mean zero random process with finite variance that is independent of $s_j$}.  Our goal is to obtain a smooth {map} $\rho \colon {L^2(\mathbb{R})} \rightarrow \mathbb{R}^{d}$ for some ${d} > 0$ such that {$\rho|_{\mathcal{M}}$ is close to an isometric embedding of $\mathcal{M}$, and}
\begin{gather}\label{about rho 2nd condition}
\Vert \rho(x_i) - \rho(x_j) \Vert \approx {d_g}(s_i, s_j)
\end{gather}
for all pairs of integers $i$ and $j$, 
where $\Vert \, \cdot \, \Vert$ denotes the {Euclidean} norm, {and $d_g$ denotes the geodesic distance induced by the Riemannian metric $g$ on $\mathcal{M}$.} {In other words, we hope to find a map $\rho$ that is robust to the inevitable noise $\varepsilon$ and is able to recover the geometric and topological information of $\mathcal{M}$.} Note that our use of the word {\em embedding} implies that the topology and differential structure of $\mathcal{M}$ has been preserved. %This mapping should be robust to the noise $\varepsilon$ in the time series, and 
{On the other hand, $d$} should ideally be small.
The low-dimensionality of {$\mathbb{R}^d$} is desirable because it guarantees that the observed dynamics can be studied from a compressed viewpoint. The robustness and metric-preserving properties of $\rho$ {together} guarantee that $\mathcal{M}$ is maintained as accurately as possible.
{With this embedding,} the dynamical process {$T$} {lying on the surface of $\mathcal{M}$ (see Figure~\ref{Figure:wave-shape oscillatory model})} 
can be studied by analyzing a collection of $d$ {time series associated with $\rho$: 
\begin{gather}
\{\mathbf{e}_k^\top\rho(x_j)\}_{j\geq 1},
\end{gather}
where $1 \leq k \leq d$, and $\mathbf{e}_k \in \mathbb{R}^{d}$ is a {$d$}-dimensional unit vector with a $1$ in the $k$-th entry}. In most cases, we cannot expect {$d$} to be the dimension of $\mathcal{M}$, and while we do not obtain the complicated dynamical process $T$, each of the {$d$} resulting time series allows us to study $T$ from a different viewpoint.}

To {achieve the above goal}, we turn to the field of unsupervised manifold learning.
We apply the well-established DM {algorithm} \cite{coifman2006diffusion} {to the set $\mathcal{X}$}.
The DM algorithm {is a} nonlinear approach to obtaining an approximate isometry of the data manifold using only a discrete set of samples. 
{From the machine learning perspective,} the DM algorithm is a kernel-based {spectral} method designed to manage abstract data sets for which only local information is known. It maintains several advantages over linear methods when dealing with nonlinear manifolds embedded in Euclidean space, {and} theoretical developments behind the DM algorithm have been extensive. 

{The design of the DM algorithm} is based on {the following fundamental result from spectral geometry. For} a smooth and closed Riemannian manifold, the eigenfunctions {and eigenvalues} of its Laplace-Beltrami operator can be used to {construct a map that} almost-isometrically embeds the manifold in {$\ell^2$} \cite{Berard_Besson_Gallot:1994}. (This mapping is called the {\em spectral embedding}.) %{The magnitudes of the eigenvalues of the Laplace-Beltrami operator have an intuitive interpretation as ``frequencies'' on the manifold, which connects the topology and geometry of a manifold via the idea of uncertainty principle \cite[Figure 3]{coifman2006diffusion}.}
{Consequently, the geodesic distance between two close points on the manifold can be well-approximated by the $\ell^2$ distance between their images in the spectral embedding.}  
Since the DM algorithm relies on discretely approximating eigenfunctions of the Laplace-Beltrami operator using eigenvectors of a matrix called the graph Laplacian (GL), the spectral convergence of the GL to the Laplace-Beltrami operator and the associated rate of convergence have been studied \cite{trillos2018error,dunson2019diffusion}. 
Moreover, {due to the intrinsic random matrix nature of the GL, DM enjoy a noise-robustness} property that is exceptionally useful for high-dimensional, real-world data \cite{ElKaroui:2010a,ElKaroui_Wu:2016b}. This random matrix property has recently been extensively studied; see \cite{ding2020phase} and the citations therein. 
{With the above facts combined, we can apply DM to robustly and accurately estimate the geodesic distance between any pair of close points on the data manifold, and this estimation is referred to as the diffusion distance. As a result, DM provide a solution to finding a map $\rho$ mentioned in \eqref{about rho 2nd condition}.} 
Other contributions are discussed in the supplementary material. 

\subsection{Diffusion maps} {In this subsection, we review the steps of the DM algorithm that we intend to apply to the set of a time series' oscillatory cycles in order to recover its associated wave-shape manifold.}
Take a point cloud $\mathcal{X}:={\{x_i\}_{i=1}^N}\subset \mathbb{R}^p$. 
Construct an $N\times N$ affinity matrix $W$ so that
{\begin{align}\label{Definition W matrix}
W_{ij}=e^{-\frac{\|x_i-x_j\|^2}{h}}
\end{align}
for all $1 \leq i, j \leq N$}, where the bandwidth $h>0$ is chosen by the user. {When} the signal-to-noise ratio is low, it is beneficial to set $W_{ii}=0$ \cite{ElKaroui_Wu:2016b}. {To} simplify the discussion, we use the radial basis function to design the affinity, but we are free to choose a more general kernel. Moreover, DM can be defined on a point cloud in a general metric space, but we focus on compact subsets of Euclidean space. 
{Given $0 \leq \alpha \leq 1$, define the $\alpha$-normalized affinity matrix by
\begin{gather}
W^{(\alpha)} = D^{-\alpha} W D^{-\alpha},
\end{gather}
where $D$ is a diagonal $N\times N$ matrix given by
\begin{equation}\label{degree_matrix1}
D_{ii}=\sum_{j=1}^NW_{ij}.
\end{equation}
By choosing $\alpha = 1$, the influence of the probability density function is negated \cite{coifman2006diffusion}. }
Define {the diffusion operator by}
\begin{align}\label{Definition:Atransition}
{A:=\left(D^{(\alpha)}\right)^{-1}W^{(\alpha)},}
\end{align}
where {$D^{(\alpha)}$} is a diagonal $N\times N$ matrix given by
\begin{equation}\label{degree_matrix}
{D^{(\alpha)}_{ii}=\sum_{j=1}^NW_{ij}^{(\alpha)}}
\end{equation}
{for all $1 \leq i \leq N$.}
We call $D^{(\alpha)}$ the {degree matrix} associated with $W^{(\alpha)}$.
{Then $A$ is diagonalizable and has real eigenvalues because it} is similar to the symmetric matrix 
\begin{gather}
{P = \left(D^{(\alpha)}\right)^{-1/2}W^{(\alpha)}\left(D^{(\alpha)}\right)^{-1/2}.}
\end{gather} Write $\phi_1, \phi_2,...,\phi_N\in \mathbb{R}^N$ for the eigenvectors of $A$, and write the corresponding eigenvalues as 
\begin{gather}
{1 = \lambda_1 \geq \lambda_2\geq\ldots\geq \lambda_N> 0.}
\end{gather} {We} know that $\phi_1 = \begin{bmatrix} 1 & \cdots & 1 \end{bmatrix}^\top\in \mathbb{R}^n$ {because the sum of each row of $A$ is one. } {Moreover,} {$\lambda_1>\lambda_2$} when the graph associated with $W^{(\alpha)}$ is connected}. {Finally, $\lambda_N> 0$ comes from the fact that the chosen kernel is positive definite (by {Bochner's} theorem \cite{Gelfand:1964}).}
{The $\alpha$-normalized DM (with diffusion time $t > 0$ and truncated to $0 < d < N$ dimensions) is defined as
\begin{equation}\label{DM}
\Phi^{(d)}_t:x_j \mapsto \left(\lambda_2^t\phi_2(j), \lambda_3^t\phi_3(j),...,\lambda_{d+1}^t\phi_{d+1}(j)\right)\in \mathbb{R}^{d},
\end{equation}
for all $1 \leq j \leq N$.} Note that {$\lambda_1$ and $\phi_1$} are ignored in the embedding {because} they are not informative. No universal rule exists guiding the choice of $d$. The choice depends on the problem at hand and can be obtained by optimizing some quantities of interest while paying mind to limitations imposed by the structure of the manifold \cite{Portegies:2015}. {We call $\|\Phi^{(d)}_t(x_i)-\Phi^{(d)}_t(x_j)\|_{\mathbb{R}^d}$ the {\em diffusion distance} between $x_i$ and $x_j$ \cite{coifman2006diffusion}.}

\subsection{Proposed DDmap algorithm for dynamics recovery}\label{Sec: proposed algorithm}
{
Our novel algorithm for exploring the dynamics encoded by a time series adhering to the wave-shape oscillatory model is as follows. {We coined the algorithm {\em Dynamic Diffusion map} (DDmap).} The algorithm is composed of four main steps, and the details of each step depend on the physiological time series that is under consideration. Suppose the time series is sampled uniformly at $f_s$ Hz for $T>0$ seconds; that is, there are in total $n:=\lfloor T\times f_s\rfloor$ sampling points. Denote the discretized time series as $\mathbf{x}\in\mathbb{R}^n$.  

\begin{enumerate}
\item
Apply any suitable beat tracking algorithm (such as \cite{malik2020adaptive} {if $f$ is an ECG signal)} to determine the locations of all oscillatory cycles in $\mathbf{x}$. Suppose there are $N$ resulting cycles. 
Denote the location (in samples) of the $i$-th oscillation as $t_i$, where $t_1<t_2<\ldots<t_N$.

\item
Extract the individual oscillations from $\mathbf{x}$. Denote the $i$-th oscillatory cycle as $\mathbf{x}_i\in \mathbb{R}^p$, where $p>0$ depends on the frequency of the oscillation. Note that these segments may overlap, in general. 
Define $\mathcal{X}:=\{\mathbf{x}_i\}_{i=1}^N\subset \mathbb{R}^p$.
Based on our model, $\mathcal{X}$ is a set of points sampled (possibly with noise) from the wave-shape manifold $\mathcal M$ associated with $\mathbf{x}$.

\item
Apply the DM algorithm to the set $\mathcal{X}$ in order to recover the manifold $\mathcal{M}$. Write $\Phi_t^{(d)} \colon \mathcal{X} \rightarrow \mathbb{R}^d$ for the truncated time-$t$ DM, where $t, d > 0$ are chosen appropriately. 

\item Recover the dynamics on the wave-shape manifold by constructing the time series $\sfrac{t_i}{f_s} \mapsto \Phi_t^{(d)}(i)$. Optionally, interpolate this time series over the duration $T$ of the original recording. Note that when we recover the wave-shape manifold (prior to visualizing the dynamics of interest), we forget the temporal relationship among points in {$\mathcal{X}$}. After the wave-shape manifold is recovered, {the} temporal relationship {is reconsidered in order to study the} dynamics. 
\end{enumerate}}
\noindent By performing these four steps, we are making estimates for the landmark locations $\{t_j\}_{j\geq 1}$, the noisy wave-shape functions $\{x_j\}_{j \geq 1}$ \eqref{noisywave}, and the wave-shape manifold $\mathcal{M}$. However, we do not explicitly estimate the wave-shape functions $\{s_j\}_{j\geq 1}$. 

We take this opportunity to elaborate on the extraction of individual oscillations from $\mathbf{x}$. One strategy is to delineate oscillatory cycles by first fixing universal constants $a, b > 0$ and then setting
\begin{gather}
\mathbf{x}_i(j+a+1) = \mathbf{x}(t_i + j)
\end{gather}
for all $-a \leq j \leq b$. Due to the non-constant period of oscillation, a choice of $\sfrac{a + b + 1}{f_s}$ which exceeds the minimum period observed in the time series will result in some $\mathbf{x}_i$ containing information from $\mathbf{x}_{i-1}$ or $\mathbf{x}_{i+1}$. If $\sfrac{a + b + 1}{f_s}$ does not meet the maximum period observed in the time series, then it may happen that $\mathbf{x}_i$ represents only a portion of the full oscillation. In general, the cycles $s_i$ and $s_{i+1}$ can overlap, which is frequently the case when analyzing ECG signals featuring cardiac arrhythmia or the majority of pulse waveforms.  Our approach will be to ignore regions in which overlapping occurs by making smaller choices of $a$ and $b$.

\section{Application -- Atrial fibrillation in the ECG signal}\label{Sect:Afib}

Atrial fibrillation (Af) is a cardiac arrhythmia associated with heart failure and stroke \cite{bordignon2012atrial}. During Af, ventricular contractions occur irregularly and lack P waves (see Fig.~\ref{Afib1}). The atria activate rapidly and persistently, generating a noise-like component in the ECG signal called the fibrillatory wave ($f$-wave). Patients with Af frequently experience ventricular ecotopy: premature heartbeats initiated by abnormal triggering points in the ventricle wall and not by the sinoatrial node.  
In the ECG signal, premature ventricular contractions (PVCs) appear to occur earlier than expected, lack P waves, and have relatively wide QRS complexes. 
With a forecasted 8 million cases in the United States by 2050, Af is the most common sustained cardiac rhythm abnormality \cite{netter3atrial,NACCARELLI20091534}. Regularly encountering Af in ECG analysis is inevitable. Modeling an ECG signal featuring Af is difficult because of the irregular heart rate and the commonly-encountered PVCs that cause transient changes in ventricular complex morphology. As such, the existing mathematical techniques {are limited} for extracting dynamical information from ECG signals featuring Af. We use this opportunity to showcase the effectiveness of the wave-shape oscillatory model; we show that the proposed algorithm for dynamics recovery can yield physiologically relevant information from an ECG signal featuring Af.  
\begin{figure}[h!]
\centering
\includegraphics[width=.86\textwidth]{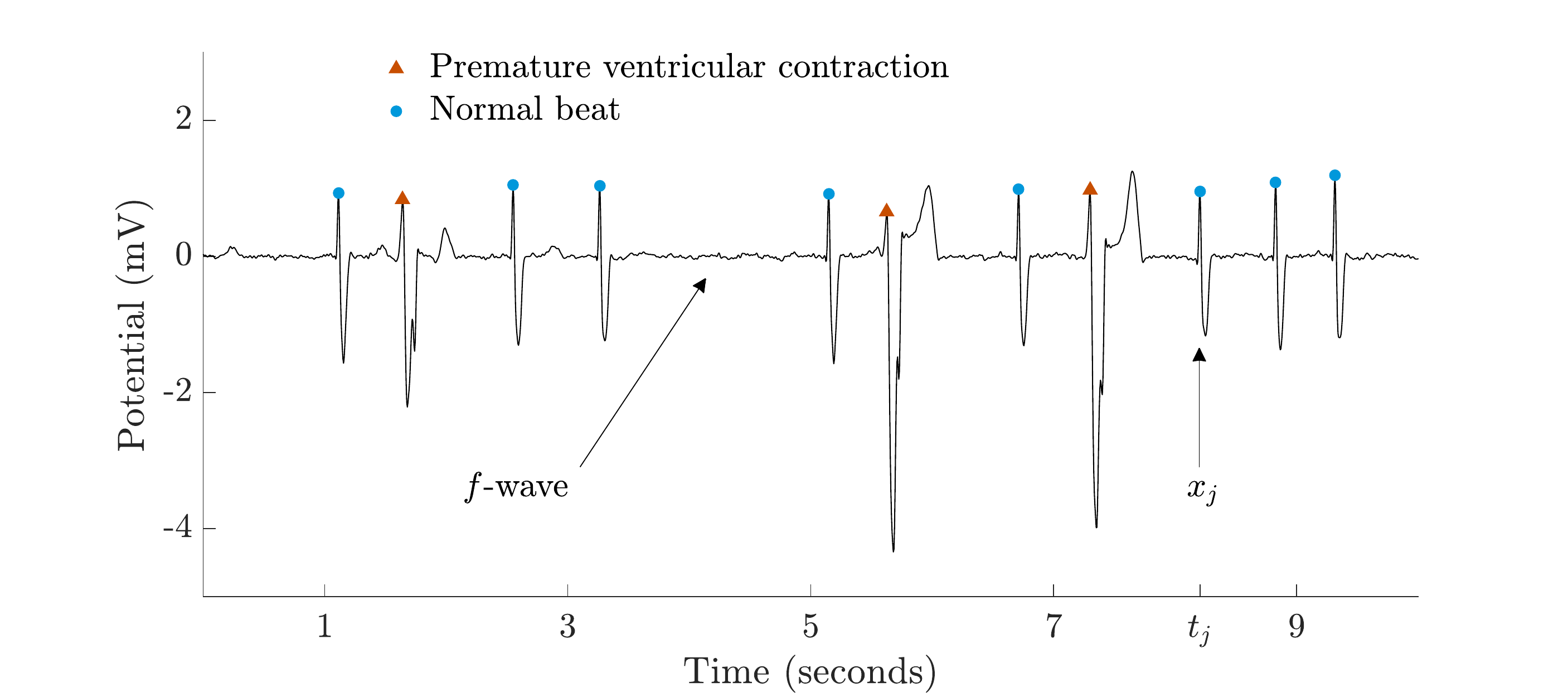}
\caption{We show an ECG signal featuring atrial fibrillation. This recording also features premature ventricular contractions (PVCs). This time series is difficult to model because the ventricular contractions appear irregularly. Moreover, the shape of each PVC is significantly different than the shape of its preceding normal beat.}\label{Afib1}
\end{figure}

\subsection{Methods and Results} 
The ECG signal used in this demonstration (of which Figure~\ref{Afib1} is an excerpt) is taken from an overnight polysomnogram (PSG) recorded from a 61-year-old male at the sleep center in Chang Gung Memorial Hospital (CGMH), Linkou, Taoyuan, Taiwan. The Institutional Review Board of CGMH approved the study protocol (No. 101-4968A3). The Alice 5 data acquisition system (Philips Respironics, Murrysville, PA) was used. Among the available channels, we consider only the first of two ECG leads and the \texttt{CFLOW} signal (a nasal cannula-based respiratory signal), which were both sampled at 200 Hz. We manually selected a 35-minute section of the recording which we determined to contain no motion artifacts.
Measurement noise and powerline interference in the ECG signal are suppressed by applying a bi-directional, third-order Butterworth lowpass filter with a cutoff frequency of $40$ Hz.  
Baseline wandering is estimated and subsequently removed using a two-step process \cite{de2004automatic}. First, we apply a median filter with a window size of 200 ms to the ECG signal. Then, we apply a median filter with a window size of 600 ms to the output of the previous median filter. 
We detect the QRS complexes in the ECG signal by applying a standard high-accuracy QRS detector \cite{malik2020adaptive}. 
Write the pre-processed ECG signal as a vector $\mathbf{x} \in \mathbb{R}^{n}$, where $n= f_s \times T $ is the number of samples, $f_s= 200$ Hz is the sampling rate of the signal, and $T = 2155$ is the duration of the recording in seconds. We find $N=2676$ R peaks in the ECG signal.
Write the location in samples of the $i$th R peak as $r_i$. 
\if false
For the purpose of analyzing the traditional EDR signal, the S peak following each R peak is determined by
\begin{gather}
s^{(k)}_i = {\arg \min}_{r^{(k)}_i +1 \leq t \leq r^{(k)}_i +60} \mathbf{E}_k( t)\, ,
\end{gather} 
where $60$ ms is chosen according to normal electrophysiology.
\fi
By defining a window surrounding each R peak that extends $80$ ms to the left and $400$ ms to the right, we are able to excise each cardiac cycle.  Define
\begin{gather}
\mathbf{X}(i,j+\lfloor 0.08 \times f_s \rfloor) := \mathbf{x}\left( r_i + j \right)
\end{gather}
for all ${-}\lfloor 0.08 \times f_s \rfloor \leq j\leq \lfloor 0.4 \times f_s \rfloor$.
Clearly, ${\mathbf{X}(i,j)}\in \mathbb{R}^{N\times p}$, where $p = 96$. 
Note that the row-ordering of $\mathbf{X}$ subliminally encodes the temporal information. 

{The DDmap} algorithm embeds {the ECG signal} (whose sampled points are the rows of $\mathbf{X}$) into $\mathbb{R}^{32}$ {via $\Phi_t^{(32)}$, where we} use a diffusion time of $t = 10$ and a normalization factor of $\alpha = 1$. The bandwidth parameter $h$ {in \eqref{Definition W matrix}} is set to be the first quartile of the set
\begin{gather}
\left\{ \sum_{k=1}^p \left[ \mathbf{X}(i, k) - \mathbf{X}(j, k) \right]^2 : 1 \leq i , j \leq N\right\}.
\end{gather}
We obtain a matrix $\mathbf{E} \in \mathbb{R}^{N \times 32}$ whose rows are the embedded cardiac cycles; that is, 
\begin{gather}
\operatorname{row}_i \mathbf{E} = \Phi_t^{(32)}\left( \operatorname{row}_i \mathbf{X} \right).
\end{gather}
To compress the available dynamical information, we obtain a low-rank approximation of $\mathbf{E}$. Specifically, we set $\mathbf{U} \in \mathbb{R}^N$ to be the top left-singular vector of $\mathbf{E}$. We can view $\mathbf{U}$ as a time series whose value at $\sfrac{r_i}{f_s}$ seconds is $\mathbf{U}(i)$. In Figure~\ref{Fig:MorphDyn}, we see that $\mathbf{U}$ encodes coarse morphological differences between heartbeats, serving as an unsupervised detector of {PVCs}. We can split the set of heartbeats into two clusters using $\mathbf{U}$.  We define
\begin{gather}
C_1 = \{ 1 \leq i \leq N : \mathbf{U}(i) \geq 0\} \quad C_2 = \{1 \leq i \leq N : \mathbf{U}(i) < 0\}.
\end{gather}
We assume that the larger set contains the normal heartbeats that were triggered by the sinoatrial node. In our case, we found that $|C_1| < |C_2|$ and that non-negative values of $\mathbf{U}$ were indicative of ventricular ectopy in the ECG signal. Ventricular ectopy detection is a well-studied discipline that aims to automate the search through ambulatory ECG recordings for irregular heartbeats known to be associated with conditions such as cardiomyopathy and congestive heart failure \cite{netter3ventricular,latchamsetty2016premature,bikkina1992prognostic,FELDMAN1970666,SWENNE1973150}. 
We visualize the clustering of heartbeats in Figure~\ref{Fig:ClusterBeats}.
\begin{figure}[h!]
\centering
\includegraphics[width=\textwidth]{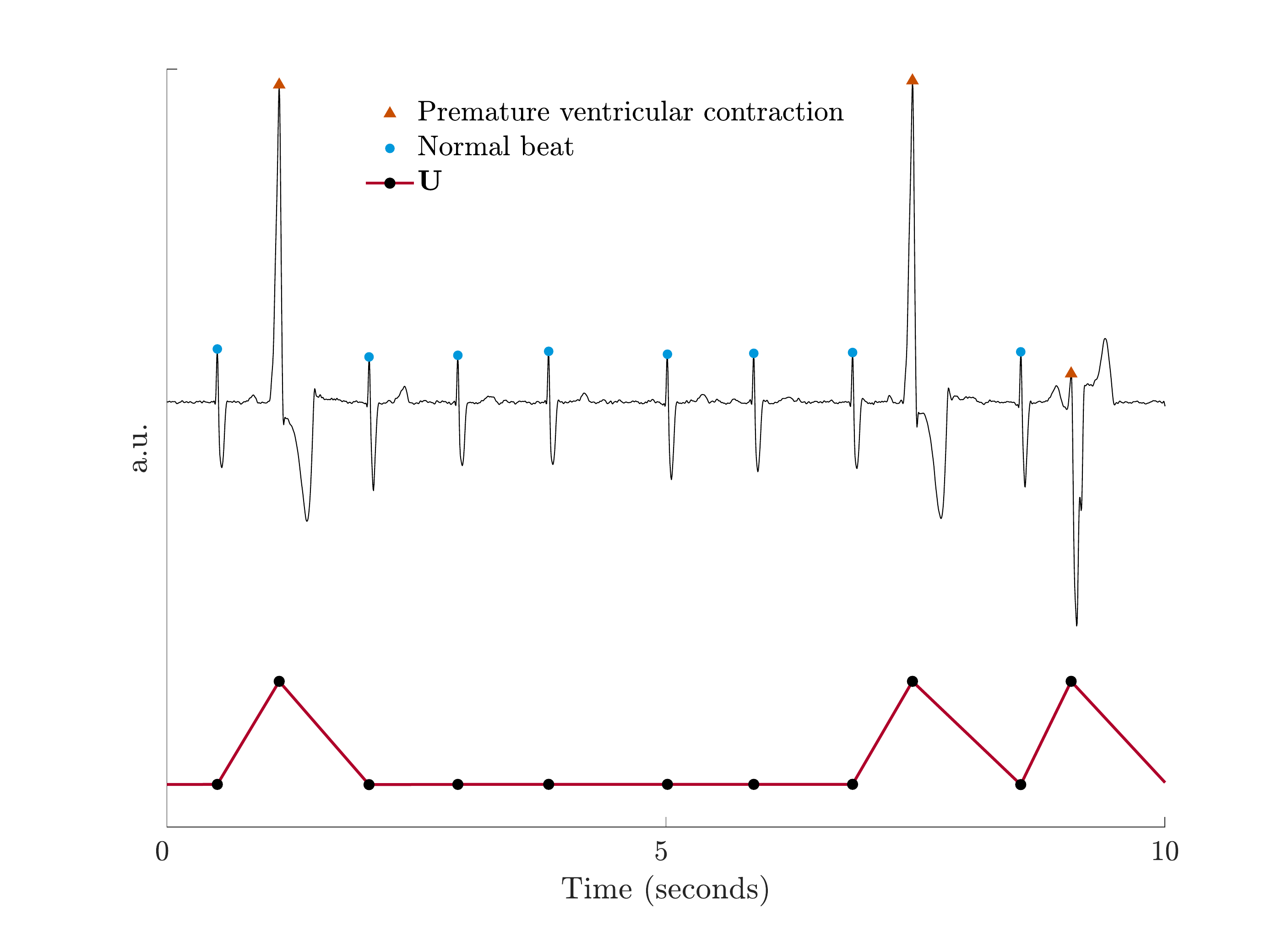}
\caption{The time series $\mathbf{U}$, representing transitions between connected components of the wave-shape manifold, corresponds to coarse changes in heartbeat morphology.}\label{Fig:MorphDyn}
\end{figure}
\begin{figure}[h!]
\centering
\includegraphics[width=\textwidth]{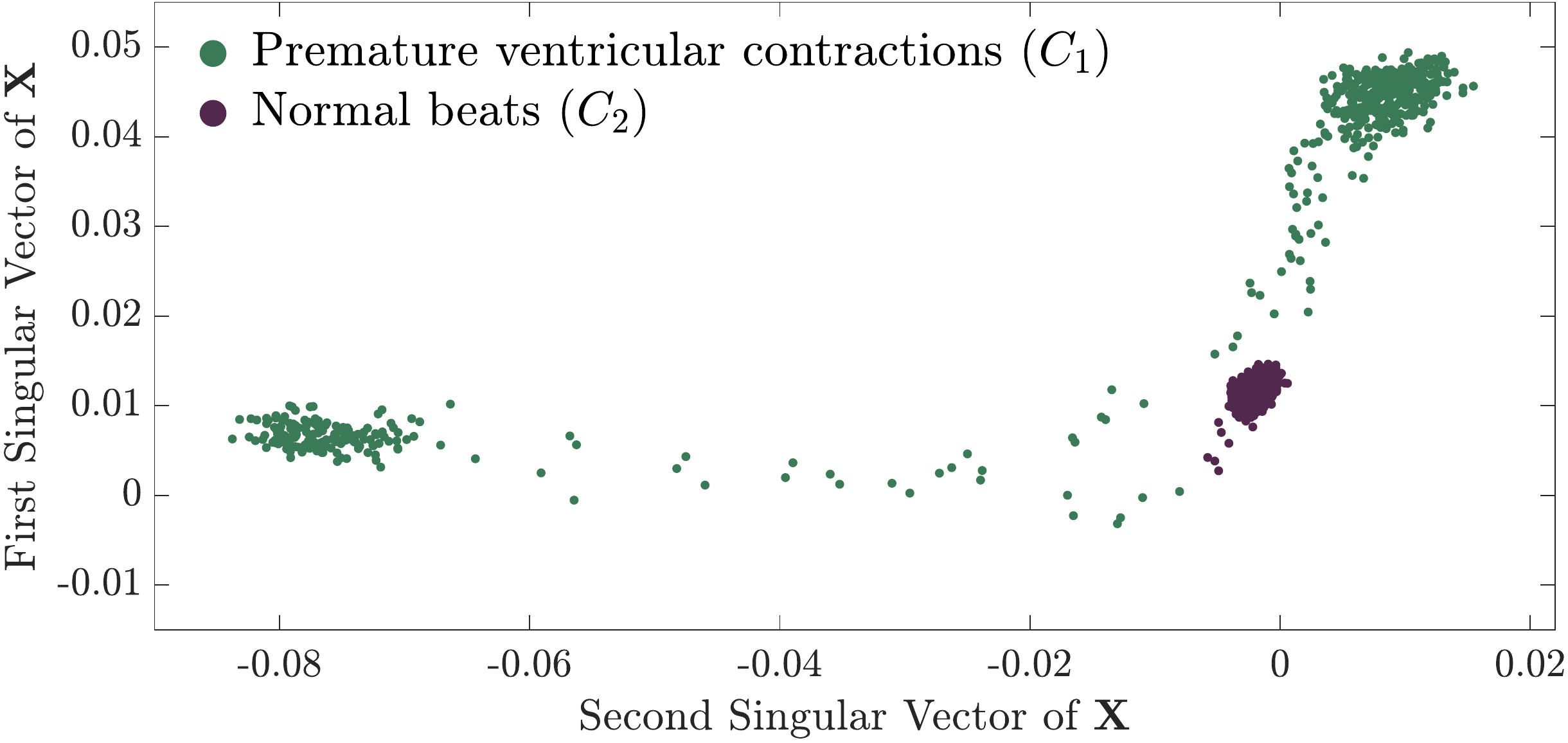}
\caption{A low-rank approximation of $\mathbf{X}$ through the singular value decomposition yields a visualization of the set of cardiac cycles in the ECG signal. Heartbeats belonging to the set $C_1$ are colored green, and those belonging to the set $C_2$ are colored purple. The set $C_2$ contains 2113 heartbeats, and the set $C_1$ contains $563$ heartbeats.}\label{Fig:ClusterBeats}
\end{figure}

We can now examine wave-shape modulation among those heartbeats not belonging to the ectopic class. We expect the normal heartbeats in the ECG signal to lie on a distinct connected component of the wave-shape manifold. A recovery of this component is simply a submatrix $\mathbf{F} \in \mathbb{R}^{|C_2| \times 32}$ of $\mathbf{E}$. 
Write $\mathbf{V} \in \mathbb{R}^{\lfloor 4T  \rfloor}$ for the piece-wise cubic spline interpolation of $\operatorname{col}_1\mathbf{F}$ to a regular sampling rate of 4 Hz, recalling that $\mathbf{F}(i, 1)$ is the value taken by the time series during the $i$-th {\em normal} heartbeat. A normalization procedure makes it easier to detect $\mathbf{V}$'s oscillatory behavior:
\begin{gather}
\widehat{\mathbf{V}}(i) := \frac{\mathbf{V}(i) - \frac{1}{21}\sum_{j=-10}^{10} \mathbf{V}(i+j)}{\sum_{k=-10}^{10} \left[ \mathbf{V}(k) - \frac{1}{21}\sum_{l=-10}^{10} \mathbf{V}(k+l)  \right]^2 }
\end{gather}
We show $\widehat{\mathbf{V}} \in \mathbb{R}^{\lfloor 4T \rfloor}$ in Figure~\ref{Fig:DynResp} above the simultaneously recorded \texttt{CFLOW} signal from the PSG. The signal $\widehat{\mathbf{V}}$ is evidently capturing the changing state of the respiratory system.  This result is to be expected; respiratory volume is known to be related to the amplitude modulation of the ECG signal. When the lung is full of air, the ECG electrode moves further from the heart, and thoracic impedance increases, causing the amplitude of the recorded ECG signal to decrease. On the other hand, when the lung is empty, the ECG electrode moves closer to the heart, and thoracic impedance decreases \cite[Chapter 8]{Clifford:2006:AMT:1213221}. This relationship between respiratory volume and the amplitude of the ECG signal has lead to the design of {\em ECG-derived respiration (EDR)} signals \cite{helfenbein2014development,varon2020comparative}. In the case of Af, these algorithms are challenged by the presence of the $f$-wave and any ectopic beats.  However, our natural, bottom-up approach based on the wave-shape oscillatory model successfully obtains the desired dynamical information from an ECG signal in the presence of both of these obstacles. 
\begin{figure}[h!]
\centering
\includegraphics[width=\textwidth]{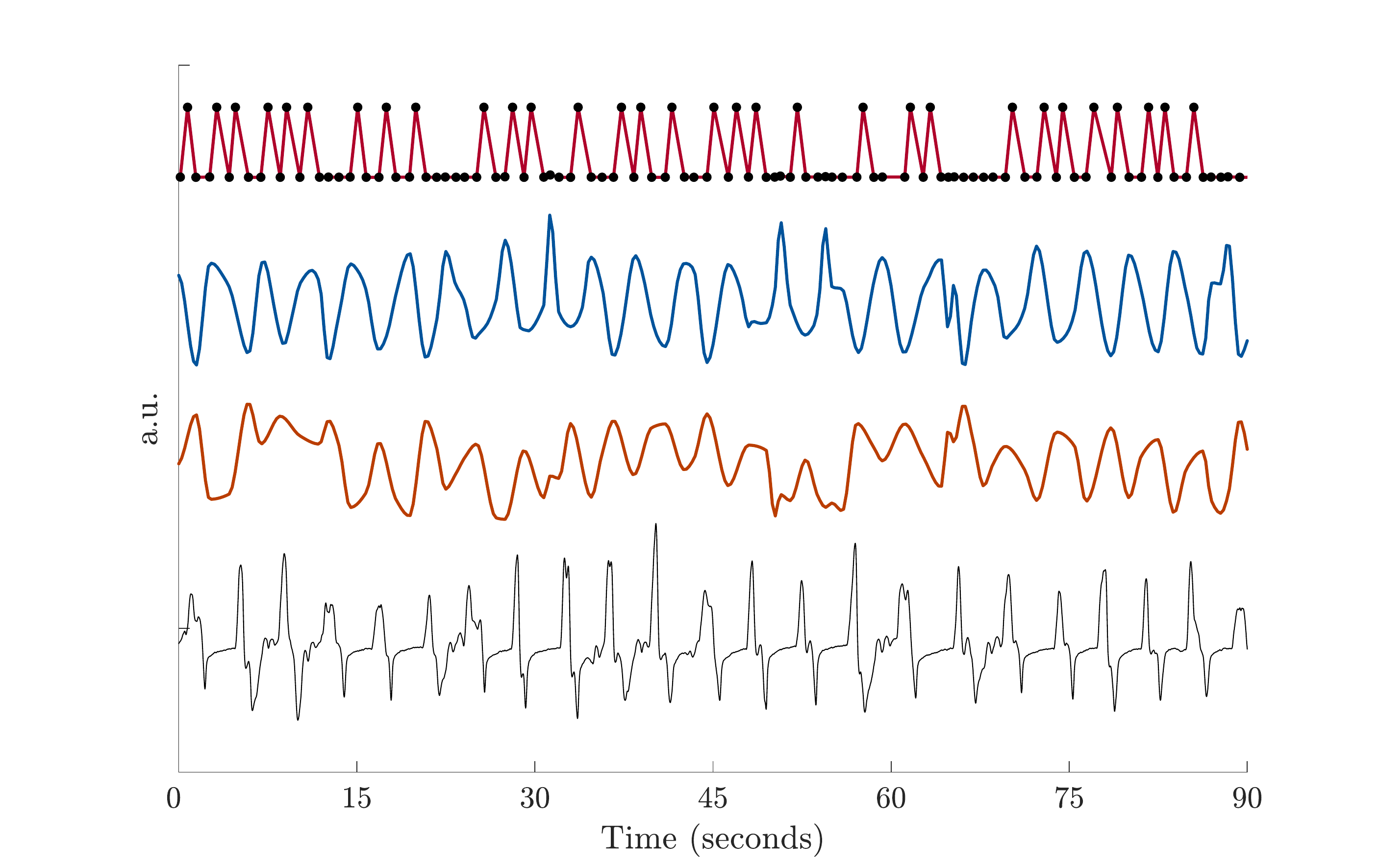}
\caption{The time series $\widehat{\mathbf{V}}$ (blue) correlates well with the simulataneously recorded respiratory signal (bottom). The traditionally acquired EDR signal (based on measuring the RS heights of non-PVC beats) shown in orange appears to provide less reliable respiratory information. For reference, we show the time series $\mathbf{U}$ (top) whose peaks indicate the presence of ectopic heartbeats in the ECG signal.}\label{Fig:DynResp}
\end{figure}

\section{Discussion}\label{sec:discussion}

In this paper, we propose a novel {{\em wave-shape manifold}} to model the time-varying morphology of cycles commonly observed in oscillatory physiological time series.
We mention that the proposed wave-shape manifold model has been {implicitly} applied in fetal ECG analysis \cite{SuWu2017}, {$f$}-wave analysis \cite{Malik_Reed_Wang_Wu:2017}, intracranial electroencephalogram analysis \cite{alagapan2018diffusion}, and {as a} novel imaging technique for long-term physiological time series \cite{Wang_Wu_Huang_Chang_Ting_Lin:2019}. This is the first time we provide a systematic discussion of the model.
This wave-shape manifold model is then used to model oscillatory physiological time series; we coin this model the {{\em wave-shape oscillatory model.}}
The wave-shape oscillatory model generalizes the traditional phenomenological model by capturing nonlinear changes in cycle morphology which occur due to all factors and not simply amplitude and frequency modulation. Moreover, it bypasses the troublesome ``slowly varying'' assumption and allows for {rapid changes in cycle shape}.
Given a time series adhering to the wave-shape oscillatory model, we  {propose a novel algorithm (DDmap) to explore the dynamics encoded in its wave-shape modulation}.
The performance and usefulness of the proposed model and algorithm are evaluated on two real physiological time series: the ECG {signal (see Section \ref{Sect:Afib})} and the ABP signal {(see Section \ref{sec:application2}} {in the supplementary material)}. We show that the proposed approach has the potential to obtain clinically relevant information from a non-traditional perspective. {We mention that the proposed model and algorithm could be applied to analyze time series outside the biomedical field.}
Last but not least, it is worth mentioning that, compared with our approach, using {contrived features based on landmarks} to quantify the oscillatory pattern is a dimension reduction step {that is not data-driven} from the machine learning viewpoint {and may} result in the loss of important information. 

\subsection{Clinical application}

In this study, we use two data sets to demonstrate the proposed model and algorithm. {First, we show that applying the proposed model and algorithm to a single-lead ECG signal simultaneously yields ventricular ectopy detection and a new {EDR} signal. This result was obtained on a time series whose frequency and wave-shape were not slowly varying and showcases the effectiveness of the wave-shape oscillatory model for handling pathological ECG signals.
The second database consists of an ABP signal recorded during the endotracheal intubation surgical procedure. Through this database, we can observe the cardiovascular effects of endotracheal intubation as documented by the ABP wave-shape. Even though patients in an unconscious state would not feel pain, the surgical step still elicits nociception in the human body and produces dynamic effects on the circulation system; these changes are encoded as subtle wave-shape modulation in the ABP signal. 
The proposed algorithm has the potential to help quantify noxious stimuli during surgery, and we will report this application in future work.}

{The proposed algorithm can be directly applied to visualize} long-term physiological time series. As technology advances, ultra-long-term monitoring systems {are growing in availability, effectiveness, and popularity} \cite{barrett2014comparison}. {However, complimentary tools for the analysis and visualization of ultra-long-term signals are lacking.} Specifically, it is difficult to directly visualize a 3-hour {long} physiological time series {(such as an ABP signal)} using the current patient {monitoring }system. Most of the time, the transition from one state to another {is delicate and slowly occurring, making it difficult to perceive by directly looking at the time series.}
Our approach, however, has the potential to reveal {long-term} dynamical information by excising the times series' cycles and embedding them into a low-dimensional space via the DM {algorithm} \cite{Wang_Wu_Huang_Chang_Ting_Lin:2019}. The result is a set of information-dense time series recorded at a much lower sampling rate than the original signal. 
Ultimately, we expect to generate a holographic visualization {tool for} long-term physiological time series. While this topic is out of the scope of the current paper, we have begun to explore it in other work \cite{Wang_Wu_Huang_Chang_Ting_Lin:2019}.

\subsection{Comparison with existing manifold models for time series}

It is important to mention that there have been several attempts to study time series by taking a manifold model into account. A well-known approach is Takens' lag map
\cite{Takens:1981}. Suppose the discrete dynamics of interest {are described by a function $X \colon \mathbb{N} \rightarrow \mathcal{M}$, where $\mathcal{M}$ is a  $d$-dimensional manifold, and that we observe these dynamics via an observational function $f\colon \mathcal M\to \mathbb{R}$; that is, we have a time series $Y = f \circ X$.} Takens' lag map is obtained by {constructing a set of vectors $\mathcal{Y}_{p,\tau} := \{y^{(\tau)}_i\}_{i\in \mathbb{N}}$ defined as
\begin{gather}
y^{(\tau)}_i := \begin{bmatrix} Y(i)& Y(i + \tau) &\cdots & Y\left(i+(p-1)\tau \right)\end{bmatrix}^\top\in \mathbb{R}^p,
\end{gather}
{for all $i \in \mathbb{N}$, where $\tau > 0$ and $p> 0$} are chosen by the user.
When $p > 2d$ and mild conditions on $\tau$, $f$, and $X$ are met, the set $\mathcal{Y}_{p,\tau}$ is related to the set $\{X(i)\}_{i\in \mathbb{N}}$ via an embedding of $\mathcal{M}$ into $\mathbb{R}^p$ \cite{Takens:1981,sauer1991,Stark_Broomhead_Davies_Huke:1997}.} 
Usually, the next step is evaluating the {topological invariants and differential structure of $\mathcal{M}$ from $\mathcal{Y}_{p,\tau}$.}
%
%{Consider $\mathcal{Y}_{p,1}$, in which case $y_i$ is understood as a ``1-dimensional patch'' \cite{peyre2009manifold}.}
%In these approaches, the collection of patches {from} a given signal ({of} any dimension) is assumed to be a manifold called the {\em patch manifold}. When the signal is an image, the manifold structure {is levied to accomplish} various image processing {tasks}.
%
In our model, however, the manifold comes from {excising} the oscillatory cycles, not Takens' lag map. Note that to recover the wave-shape manifold, we collect oscillatory cycles by taking the temporal location of {landmarks} into account. Due to the {inevitable} non-constant {instantaneous frequency, these} landmarks do not appear regularly {in the time series}. {Assuming the extracted cycles are vectors} of length $p$, {they} form a strict subset of $\mathcal{Y}_{p,\tau}$, {one which is unobtainable using} Takens' lag map. %{If the period of oscillation happens to be a constant $L$ samples, then} the set of oscillatory cycles is exactly $\mathcal{Y}_{p,L}$. 
{We mention that if we view each oscillatory cycle as one sampling point, it is possible to combine the idea of Takens' lag map with the DDmap algorithm.}

\subsection{Future work}

The proposed model opens several research directions. 
%
%First,  On the other hand, 
{First, we need to estimate $T$, and hence forecast the time series. To this end, we need to} systematically model and explore the {intertwined processes generating $\{t_j\}_{j\geq 1}$ and $\{s_j\}_{j \geq 1}$} in \eqref{Introduction:ANHM1} {that would allow a full quantification of the underlying system and aid in predicting future behavior}. An exploration depends on the signal type and the clinical problem of interest, {and we will approach these problems in} future work.

{From the algorithmic perspective}, {our algorithm requires us to determine at least one landmark for the oscillatory pattern in order to delineate cycles and construct the point cloud associated with the wave-shape manifold.} However, {doing so} might not always be possible. When the {instantaneous frequency} is slowly varying, we may count on the de-shape short-time Fourier transform \cite{lin2018wave} to {demarcate} the cycles. In general, {the problem of cycle delineation is challenging and might depend on the signal type and application. {We mention that it is even possible to directly view slices of the spectrogram or other frequency domain distributions as wave-shape functions of an alternative form. This approach would be helpful for signals like the electroencephalogram (EEG) \cite[Chapter 4]{malik2020geometric}.}
Moreover, while the wave-shape oscillatory model and the proposed {DDmap} algorithm have been shown to work well for several physiological time series, {our approach to fully quantifying wave-shape modulation depends on the assumption that either oscillatory cycles do not overlap, or that we do not need the information contained in the overlapping regions.} We may need a better algorithm when the information {contained} in the overlapping {regions} is critical.}

Third, from the theoretical perspective, the spectral embedding and spectral convergence results we count on to validate the DM algorithm are not complete. For example, while the spectral embedding result helps explain how the DM works, it does not fully explain if, and when, the DM can achieve dimension reduction. Specifically, to the best of our knowledge, there is no result discussing how to estimate $N_E$ {(c.f. \cite[Theorem 5.1]{Portegies:2015})} {from what is known or assumed about the} geometric and topological profile of the manifold. Without this piece of information, we cannot even estimate the number {$d$} \eqref{DM} of eigenvectors needed to reconstruct the manifold.
{To the best of our knowledge}, this direction is relatively empty and needs more exploration.

\section*{Acknowledgments}
The authors gratefully acknowledge Dr. Cheng-Hsi Chang for sharing the database demonstrated in Section \ref{sec:application2} and Dr. Yu-Lun Lo for sharing the database demonstrated in Section \ref{Sect:Afib}. The work of Yu-Ting Lin was supported by the National Science and Technology Development Fund (MOST 107-2115-M-075-001) and the LEAP@Duke program of the Ministry of Science and Technology (MOST), Taipei, Taiwan. 
Hau-Tieng Wu acknowledges the hospitality of the National Center for Theoretical Sciences (NCTS), Taipei, Taiwan during the summer of 2019.

\bibliographystyle{amsplain}

\bibliography{WaveShapeDynamics}

\newpage
\begin{center}
{\large\bf SUPPLEMENTARY MATERIAL}
\end{center}

\section{Application -- Pulse waveform analysis during general anesthesia}\label{sec:application2}

We use a clinical data set to study {the physiological responses} of the cardiovascular system to surgical events. To this end, we analyze {pulse wave-shapes} obtained from the arterial blood pressure (ABP) signal. We study a particular surgical event, namely endotracheal intubation{, in which a tube is fed through the mouth and into the trachea in order to facilitate ventilation of the lungs.} It is well known that this procedure elicits nociception and a response by the autonomic nervous system (ANS). This response includes a change in heart rate and an elevation in both blood and pulse pressure. The change in heart rate corresponds to frequency modulation {(FM)} in the ABP signal, the change in blood pressure corresponds to {a modulating} trend, and the change in pulse pressure corresponds to amplitude modulation (AM). We are interested in detecting wave-shape modulation independently of AM, FM, and trend modulation.  {We will detect and quantify wave-shape modulation by applying the wave-shape oscillatory model and the associated algorithm for dynamics recovery.}

\subsection{Material}

The {ABP signal used in this experiment} was collected in the operating room in Shin Kong Wu Ho-Su Memorial Hospital, Taipei, Taiwan. With the goal of analyzing the dynamic response of the human body to distinct surgical steps, its parent database's collection was approved by the local institutional ethics review boards (Shin Kong Wu Ho-Su Memorial Hospital, Taipei, Taiwan; IRB No.: 20160106R and 20160706R), and written informed consent was obtained from each patient. By using standard patient monitoring instruments {(Philips IntelliVue\texttrademark)}, the data was collected via the third-party software, ixTrend\texttrademark Express ver. 2.1 (ixellence GmbH, Wildau, Germany). The sampling {rate} of the ABP {channel} was 125 Hz. The moment of each surgical step was registered using purpose-made software so that precision of the registration time was less than one second.  

\subsection{Data preparation} 

{First, we upsample the ABP signal (by zero-padding the signal's Fourier transform) to a sampling rate of 2000 Hz to enhance the resolution. The ABP signal is truncated so that it begins 30 seconds before the intubation event and is 360 seconds (6 minutes) in duration. We estimate and remove the trend from the ABP signal by applying a median filter with a window length of $2$ seconds. This step removes the blood pressure information.}
To obtain all pulses from the ABP signal, we detect a landmark on each cycle associated with the {\em pulse wave arrival time} {using a modified version of a systolic wave detection algorithm for photoplethysmography \cite{elgendi2013systolic}}. This landmark is determined by finding the maximum of the first derivative during the ascending phase of the ABP pulse; that is, the point at which the signal ascends fastest.
Suppose the ABP time series is discretized as $\mathbf{A} \in \mathbb{R}^{n}$, where $n=\lfloor f_s \times T\rfloor$, $f_s=2000$ is the sampling rate in Hertz, and $T=360$ is the duration of the recording in seconds. {To remove false pulses and pulses that are corrupted by noise, the $i$-th pulse is discarded if the number of significantly wide local maxima between the $i$-th and $(i+1)$-th pulse wave arrival times is larger than three. In practice, a more sophisticated signal quality index can be considered \cite{elgendi2016optimal}.}  
Suppose that there are $N$ detected pulse arrival points remaining, {and that} their locations {(in samples)} are $\{a_i\}_{i=1}^N$. The interval between two consecutive oscillations is $L_i:=a_{i+1}-a_i$, where $1 \leq i \leq N-1$.  Set $L:=\min_{i=1}^{N-1}L_i$. 
Delineate the cycles in the ABP signal as
\begin{gather}
{\mathbf{X}_i(\lfloor f_s \times 0.08 \rfloor + j+1) := \mathbf{A}( a_i + j ),}
\end{gather}
where $i=1,\ldots,N$ and $-\lfloor f_s \times 0.08 \rfloor \leq j \leq L -\lfloor f_s \times 0.08 \rfloor$. Clearly, $\mathbf{X}_i \in \mathbb{R}^p$, where $p=L+1$. 
Next, we remove the {pulse} pressure information by normalizing each $\mathbf{X}_i$. Set
\begin{equation}
{\overline{\mathbf{X}}_i(j)}:=\sigma_i^{-1}(\mathbf{X}_i(j) - \mu_i)
\end{equation}
for $j=1,\ldots,p$, where $\mu_i$ is the mean of $\mathbf{X}_i$ and $\sigma_i$ is the standard deviation.
Finally, we assemble the collection of ABP pulses as
$\mathcal{X}_{\mathbf{A}} := {\{\overline{\mathbf{X}}_i\}_{i=1}^N}$.
The point cloud $\mathcal{X}_{\mathbf{A}}$ encodes the morphological information of the ABP waveform{,} and the temporal information is preserved in the sequence {$\{a_i\}_{i=1}^N$}. 

\subsection{Result}

We apply the DM algorithm to the set $\mathcal{X}_{\mathbf{A}}$. We pick $\alpha = 1$ and a diffusion time of $t = 1$. We select the bandwidth $h$ to be the 25th percentile of the set $\left\{ \Vert v - v_{40} \Vert^2 : v \in \mathcal{X}_{\mathbf{A}} \right\}$,
where $v_{40}$ represents the $40$-th closest pulse to $v$ from the set $\mathcal{X}_{\mathbf{A}}$ in terms of the Euclidean distance. 
The first two coordinates of the DM provide an intuitive visualization. 
{In Figure \ref{Figure:abp3d_embs},} we show the embedding {of the cycles from the} signal {in} Figure \ref{Figure:pulse}.
In this ABP recording, the stimulus event, endotracheal intubation, occurred around the 30-second mark.
To read the physiological evolution elicited as nociception, we read the trajectory shown in Figure \ref{Figure:abp3d_embs}, {wherein the embedded cycles are colored according to their time of occurrence ($\sfrac{a_i}{f_s}$).}
This subject's trajectory is evidently non-linear.

Before the noxious stimulus, the pulse waveform is stable. The dark blue points in the embedding represent the patient's baseline pulse wave-shape. In response to the stimulus event, the pulse rate increases, causing the width of each wave-shape to decrease. Moreover, a prominent peak appears between the systolic peak and the dicrotic notch (see Figure~\ref{Figure:StackedPulses}). These abrupt and relatively significant changes are reflected in the large distance between the blue points and the cyan points in the embedding. By the 150-second mark, the pulse rate has returned to its baseline value (see Figure~\ref{Figure:ABP_HeartRate}), but the prominent peak between the systole and the dicrotic notch remains. The heart rate then begins to descend below its baseline value, and the prominent peak is shown to depress by the end of the recording.
This finding indicates a nonlinear relationship between the wave-shape and the heart rate. 
It is worth mentioning that since the pulse pressure information has been removed (through normalization), the embedding only reflects the dynamics encoded in {the wave-shape modulation and the frequency modulation}. 

This finding coincides with our physiological knowledge: the stimulus invokes a series of physiological responses {that include} elevated heart rate, vasoconstriction, increased heart contractility, and their subsequent interactions. The impact of the stimulus decays after a while, and the physiological status gradually stabilizes. 
%In the clinic, features such as the Surgical Pleth Index (SPI) use the pulse waveform to monitor nociception \cite{}.  
%
In {summation},
this embedding helps visualize {the cardiovascular system's} transient {dynamical response} to the stimulus event {(endotracheal intubation)}{: a combination of vascular wall tension, blood volume, and heart contractility modulations that in turn modulate the pulse wave-shape. Validating the potential of this algorithm for quantifying nociception in the operating room will be investigated in future work.}

\begin{figure}[htb!]
\centering
\includegraphics[width=.8\textwidth]{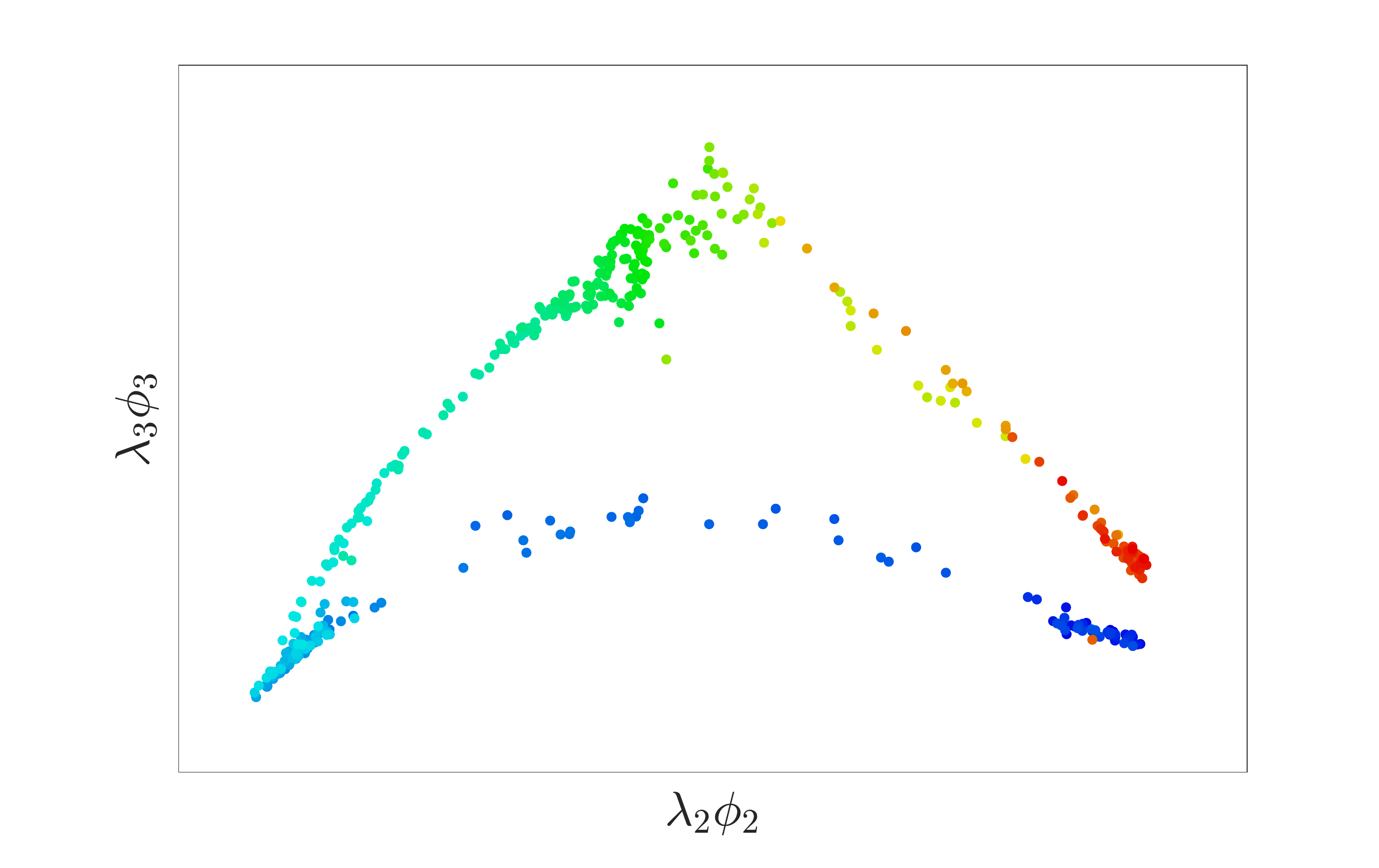}\\
\includegraphics[width=.8\textwidth]{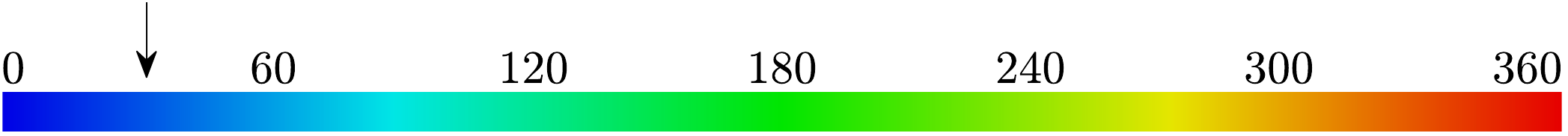}\\
Time (seconds)\\
\caption{We show a two-dimensional embedding of the pulse wave-shapes extracted from a $360$-second arterial blood pressure (ABP) signal. The signal was recorded from a patient undergoing general anesthesia before and after an endotracheal intubation event. The ABP signal is shown in Figure \ref{Figure:pulse}.
The intubation event takes place at the 30th second, indicated by the black arrow on the colorbar.}\label{Figure:abp3d_embs}
\end{figure}

\begin{figure}[htb!]
\centering
\includegraphics[width=.32\textwidth]{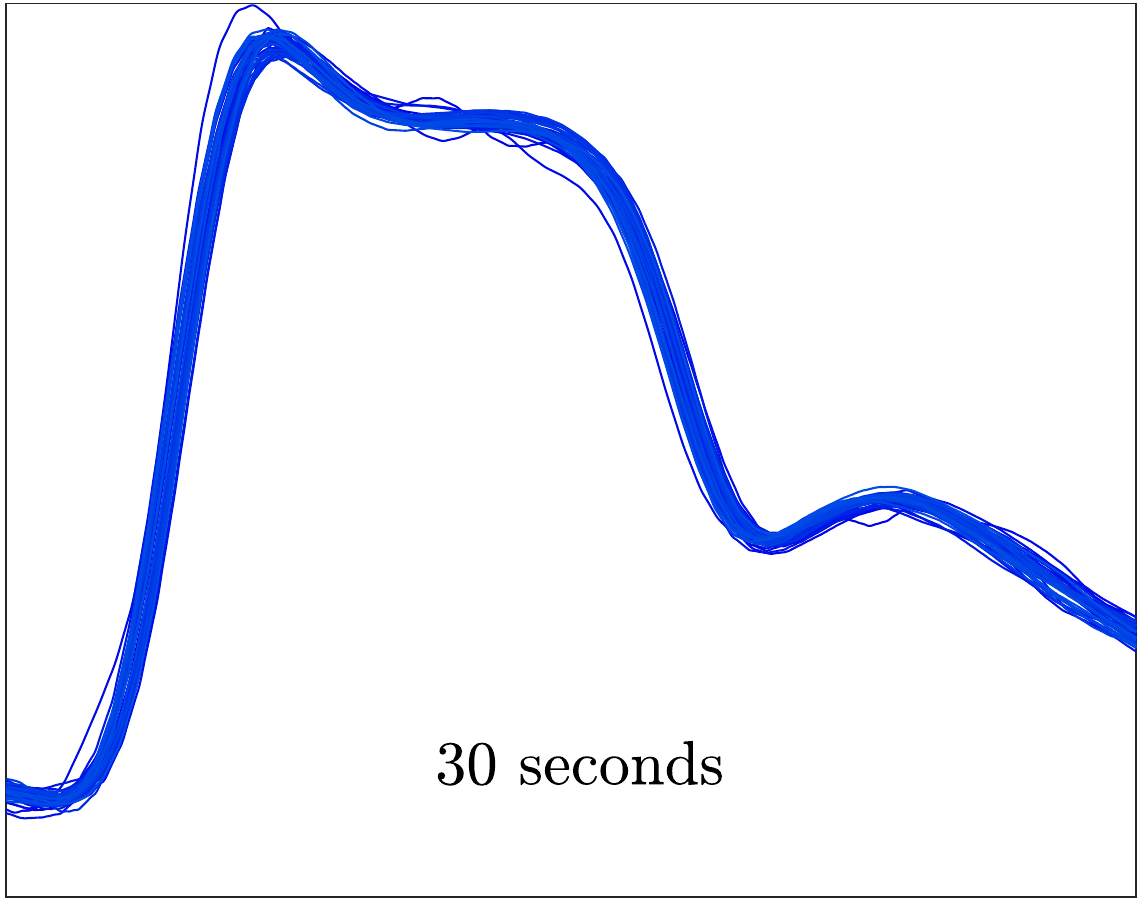}
\includegraphics[width=.32\textwidth]{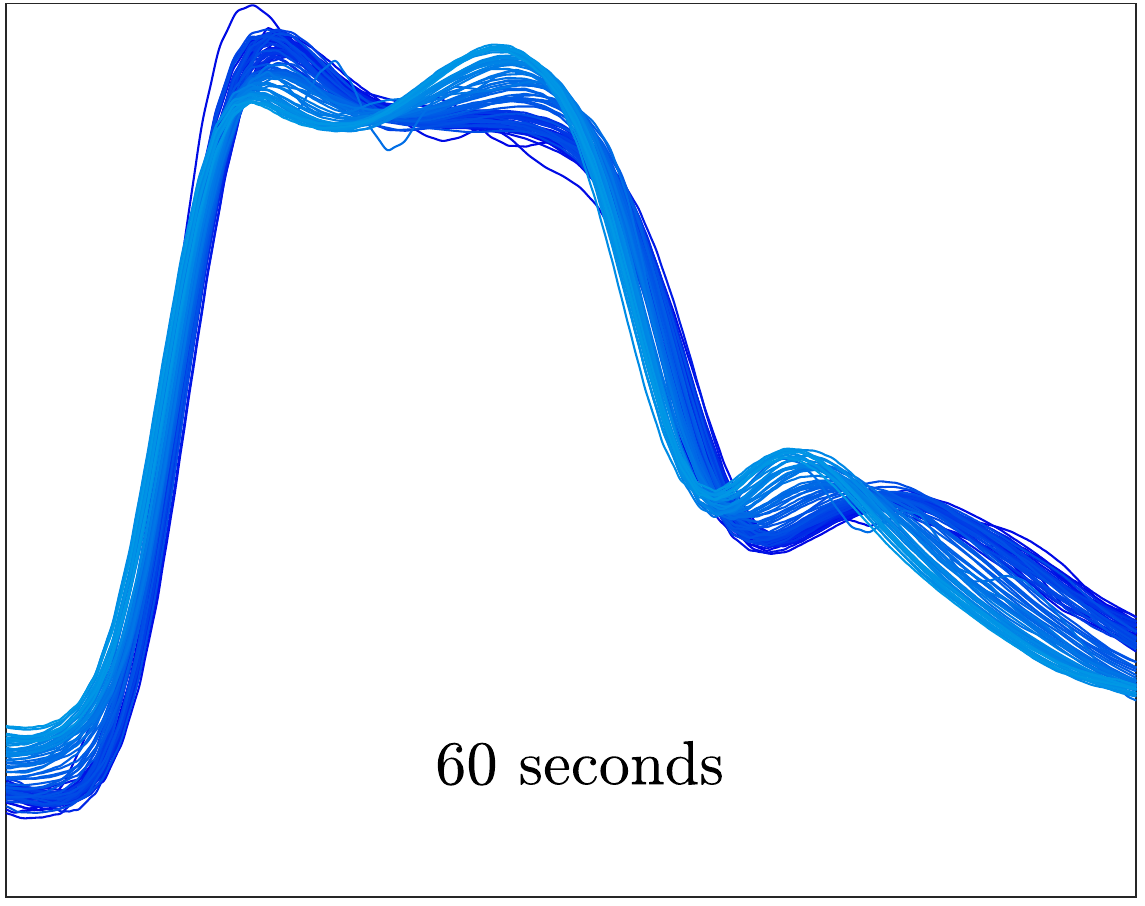}
\includegraphics[width=.32\textwidth]{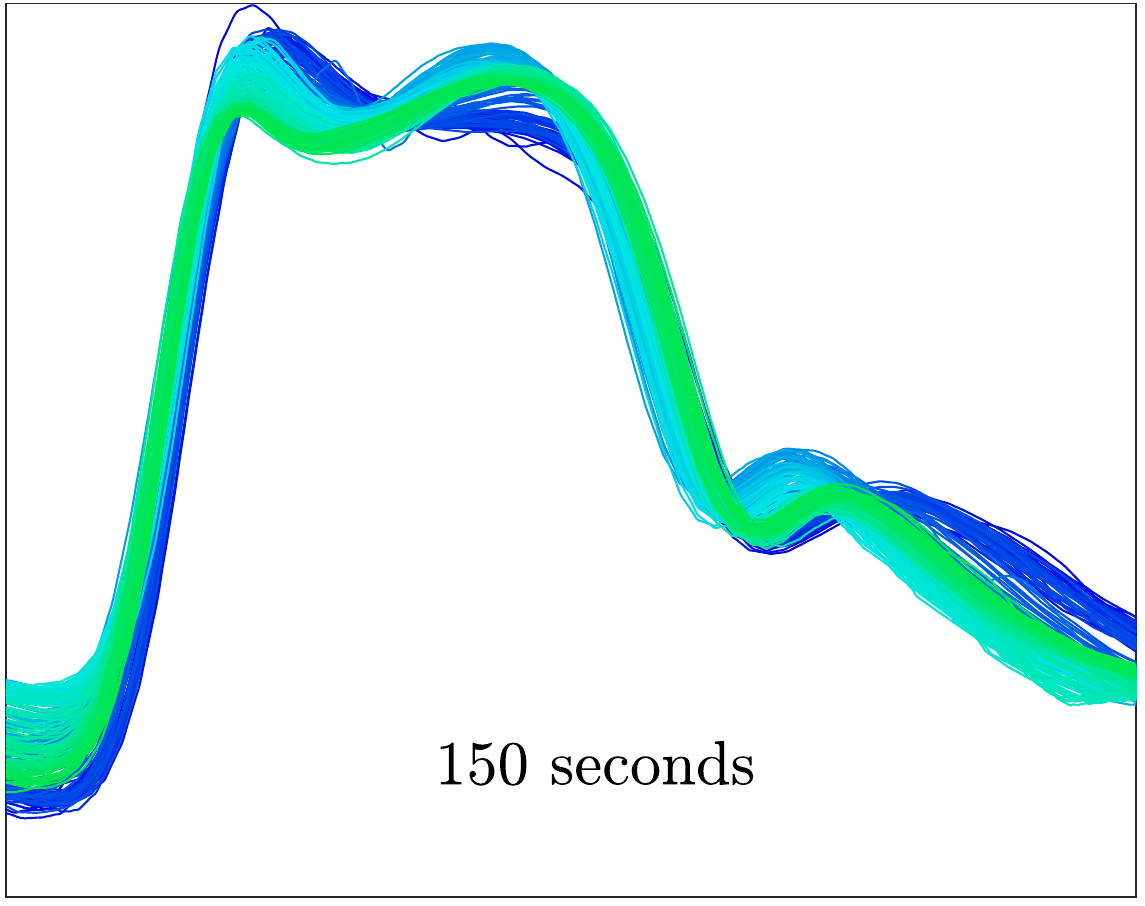}\\
\includegraphics[width=.32\textwidth]{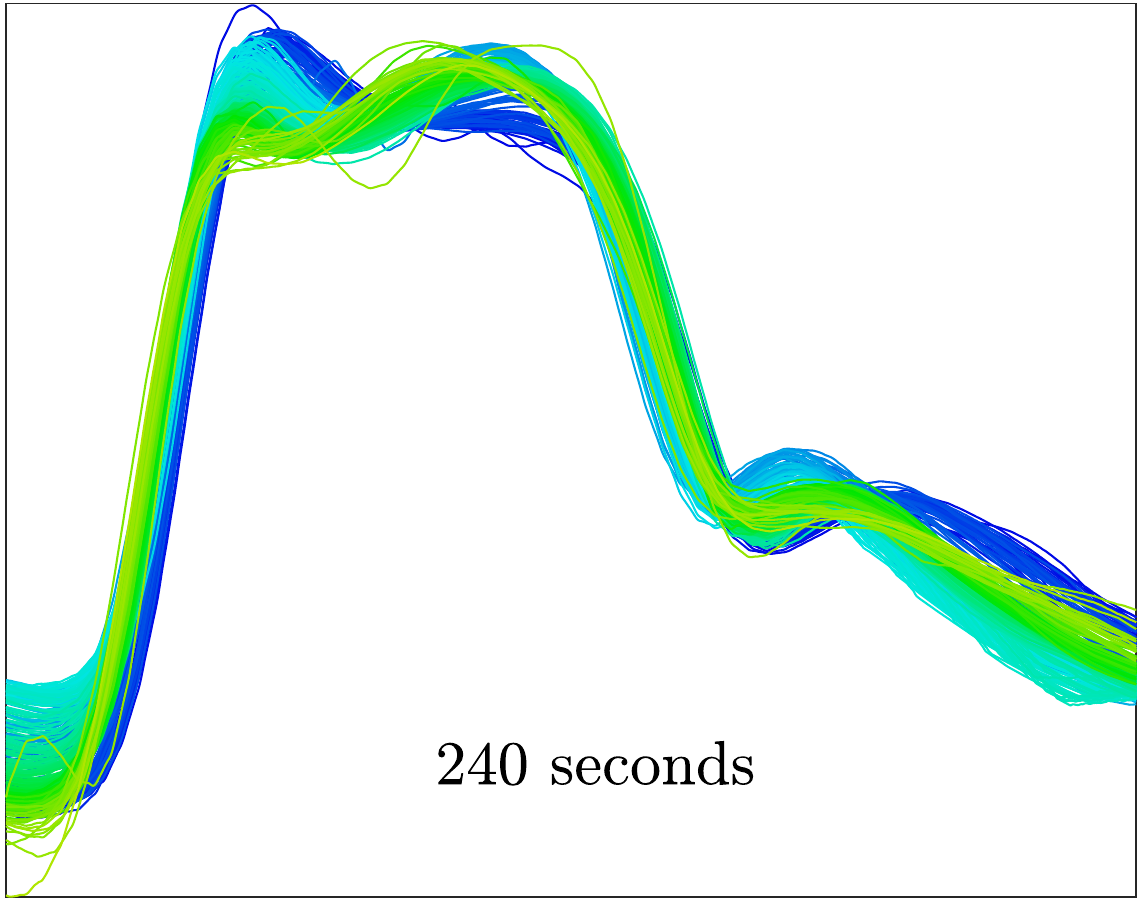}
\includegraphics[width=.32\textwidth]{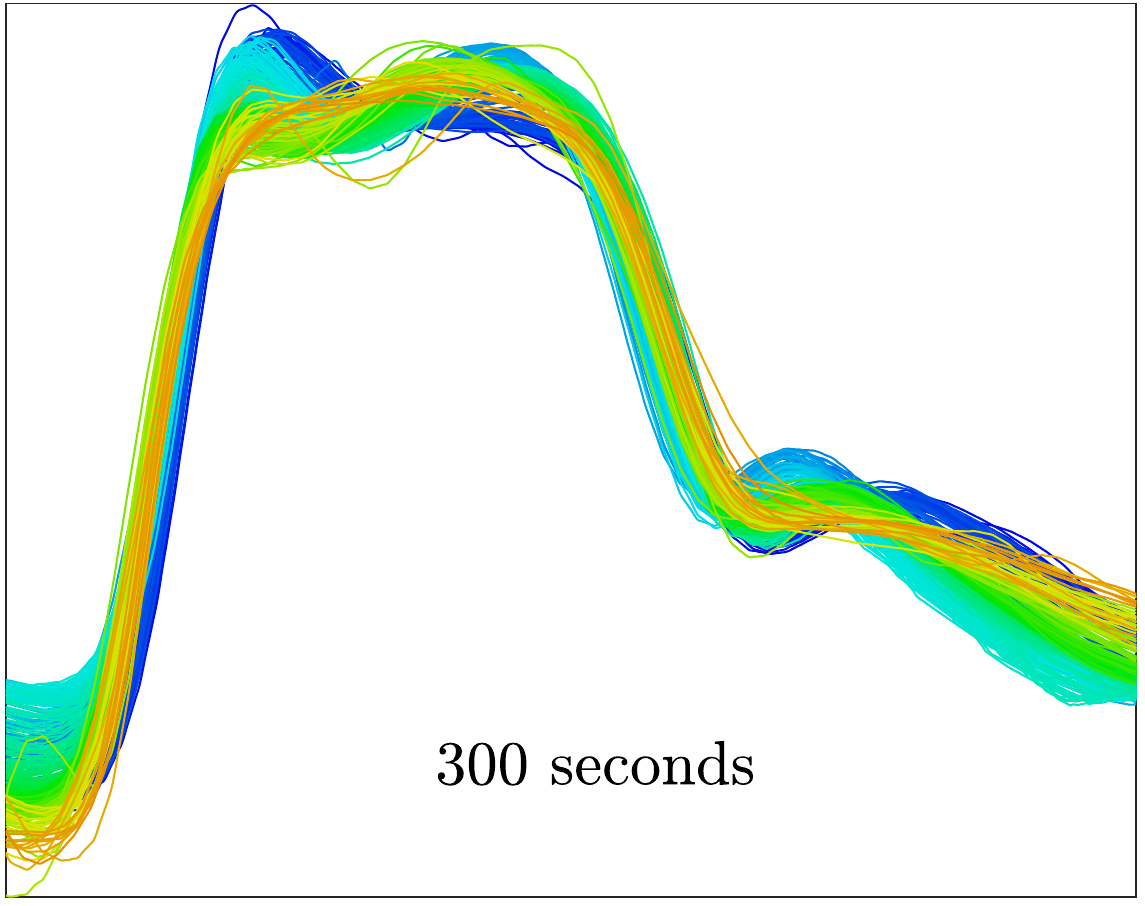}
\includegraphics[width=.32\textwidth]{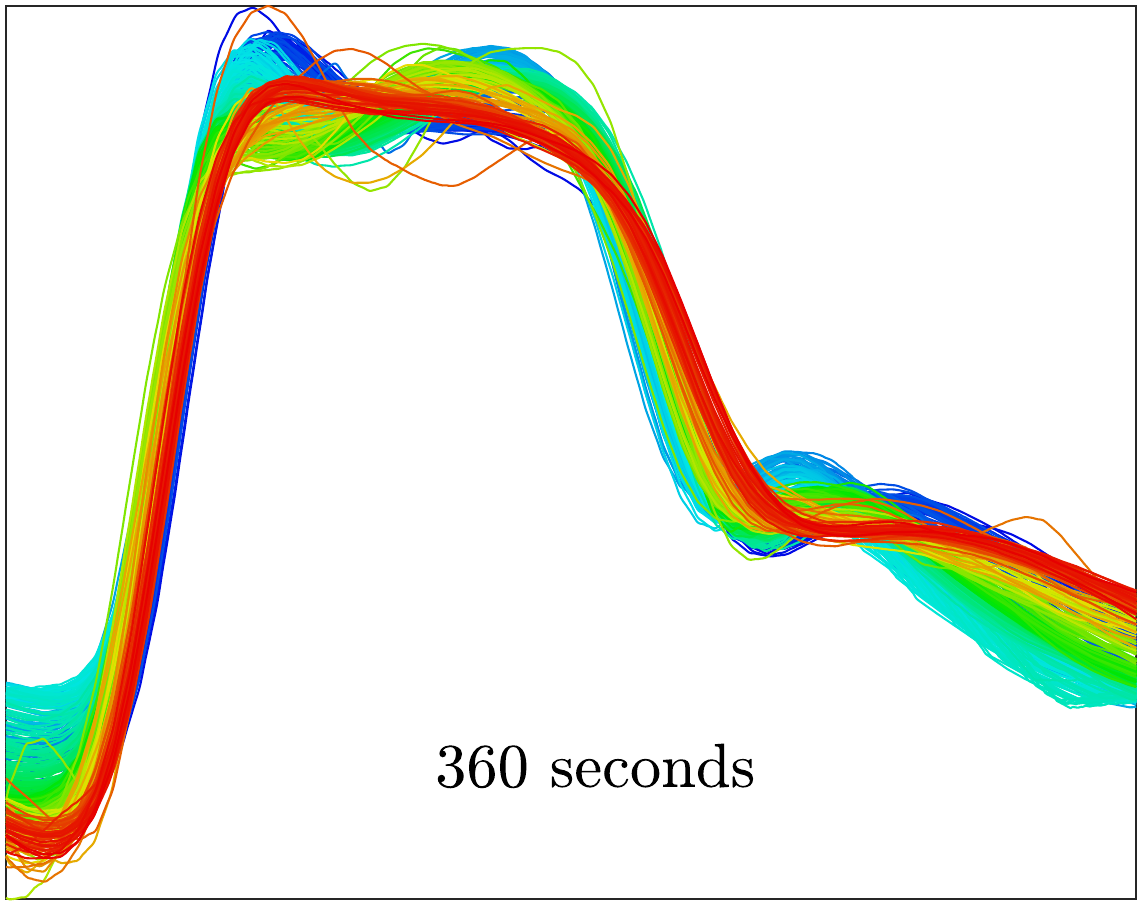}\\
\vspace{0.1in}
\includegraphics[width=.8\textwidth]{ABPColorBar.pdf}\\
Time (seconds)\\
\caption{We superimpose the normalized pulse wave-shapes from the set $\mathcal{X}_{\mathbf{A}}$. After normalization, only frequency (width) and wave-shape modulation effects are perceivable. After the stimulus event at the 30-second mark, the pulses become less wide, and a prominent peak protrudes between the systolic peak and the dicrotic notch. These effects then continue to subtlely evolve throughout the recording.}\label{Figure:StackedPulses}
\end{figure}

\begin{figure}[htb!]
\centering
\includegraphics[width=.8\textwidth]{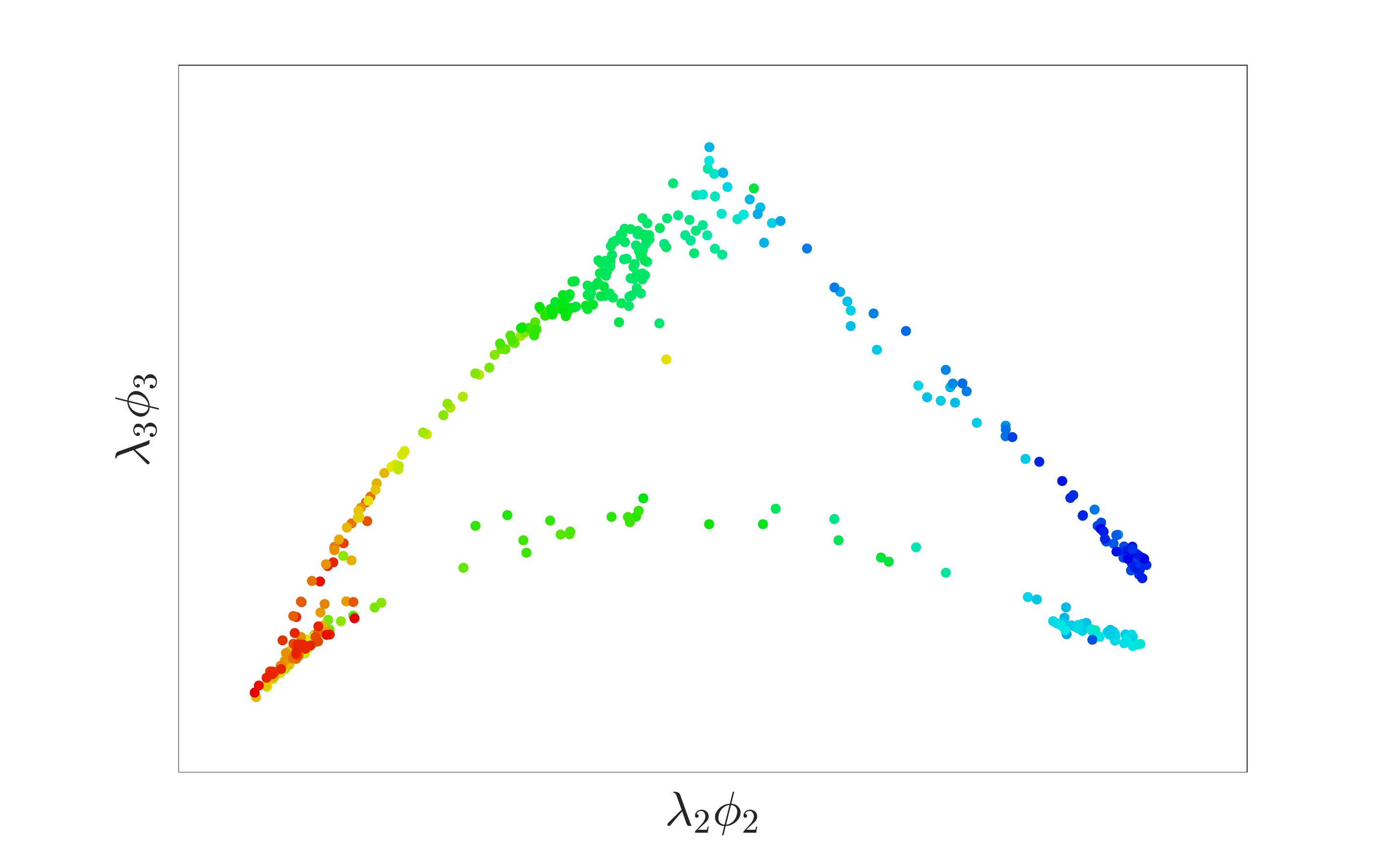}\\
\includegraphics[width=.8\textwidth]{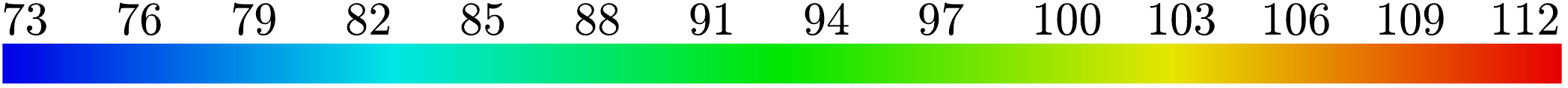}\\
Pulse rate (bpm)\\
\caption{We show the same two-dimensional embedding depicted in Figure~\ref{Figure:abp3d_embs} with a different coloring scheme. We color the $i$-th pulse according to the quantity $\sfrac{60f_s}{a_i - a_{i-1}}$, known as the instantaneous pulse rate (in beats per minute). The pulse rate surges in response to the stimulus and then descends below the original resting rate.} \label{Figure:ABP_HeartRate}
\end{figure}

\section{Previous mathematical models for non-seasonal periodic time series}\label{Section:Review}
The wave-shape oscillatory model takes an approach to modeling time series like the ECG and ABP signals that generalizes previous mathematical models.

\subsection{The phenomenological model}
To begin describing these previous mathematical models and their relationships to the wave-shape oscillatory model, consider a time series of the form
\begin{gather}
f(t) = s(t) \quad t \in \mathbb{R},
\end{gather}
where $s \in \mathcal{C}^{1,\alpha}(\mathbb{R})$ is $1$-periodic (in the sense that $s(t) = s(t+1)$ for all $t \in \mathbb{R}$) and possibly non-sinusoidal. {Here we assume some regularity of $s$ to simplify the discussion. It can be relaxed by taking the distribution theory into account.} In particular, the function $s$ restricted to the interval $[-\sfrac{1}{2}, \sfrac{1}{2}]$ describes the shape of the cycles appearing in the time series. This simple model is yet insufficient for describing physiological time series because physiological time series almost always vary in frequency and amplitude over time (due to the dynamics of the underlying system).  The {\em phenomenological model} proposed in \cite{wu2013instantaneous,hou2016extracting,xu2018recursive} takes amplitude and frequency modulation into account.  
The oscillatory time series $f \colon \mathbb{R} \rightarrow \mathbb{R}$ is modeled as 
\begin{align}\label{Introduction:ANHM0}
f(t)=a(t)s(\phi(t))+T(t)+\Phi(t) \quad t \in \mathbb{R},
\end{align}
where
\begin{itemize}
\item $a\in C^1(\mathbb{R})\cap L^{\infty}(\mathbb{R})$ is a positive, smooth function describing amplitude modulation;
\item $s\in C^{1,\alpha}(\mathbb{R})$ is $1$-periodic and describes the shape of the oscillations in the time series (coined the {\em wave-shape function} \cite{wu2013instantaneous});
\item $\phi\in C^2(\mathbb{R})$ is a smooth, monotonically-increasing function describing phase modulation ($\phi' \in L^{\infty}(\mathbb{R})$ is termed the {\em instantaneous frequency} of $f$);
\item $T\in C^1(\mathbb{R})$ is the {\em trend} that intuitively is locally constant \cite[(6)]{Chen_Cheng_Wu:2014};
\item $\Phi \colon \mathbb{R} \rightarrow \mathbb{R}$ models the inevitable random noise, which we assume to be stationary. 
\end{itemize}
In this work, we always assume that the trend has been removed. For this reason, we do not mention a trend term in the definition of the wave-shape oscillatory model.
Since $a$ and $\phi'$ are not required to be constant, this model is suitable for describing physiological time series that change in amplitude and frequency over time (with some degree of regularity). 
Under this phenomenological model, the usual mission in time series analysis is estimating  $a$, $s$, $\phi$, and the statistical properties of $\Phi$ from one realization of $f$.
 To ensure identifiability, we need to assume that $a$ and $\phi'$ are {\em slowly varying}. To be specific, we say that $a$ and $\phi'$ are slowly varying with parameter $\epsilon > 0$ if the following conditions hold.
\begin{itemize}
\item $|a'(t)|\leq \epsilon \phi'(t)$ for all $t\in \mathbb{R}$;
\item $|\phi''(t)|\leq \epsilon \phi'(t)$ for all $t \in \mathbb{R}$.
\end{itemize}
When we recover $a$ and $\phi$, we are recovering dynamical information about the system being monitored in a limited (but in many cases sufficient) sense. When we recover $s$, we are recovering a representation of the ``mean'' state for the system.  This statement will be made more precise later, but it corresponds well with traditional approaches to cardivascular waveform analysis. 
We now provide some specific examples of how $a$ and $\phi$ encode the changing physiological state of the system being monitored. 

\begin{example} The instantaneous amplitude of a single-lead ECG signal is directly related to respiratory volume (see Figure~\ref{edrillustration}).  When the lung is full of air, the ECG electrode moves further from the heart, and thoracic impedance increases, causing the amplitude of the recorded ECG signal to decrease. On the other hand, when the lung is empty, the ECG electrode moves closer to the heart, and thoracic impedance decreases \cite[Chapter 8]{Clifford:2006:AMT:1213221}. This relationship between respiratory volume and the amplitude of the ECG signal has lead to the design of algorithms which estimate the respiratory signal from the ECG signal. The estimated respiratory signal is called the {\em ECG-derived respiration (EDR)} signal.  We {have discussed} EDR in Section~\ref{Sect:Afib}.
\end{example}

\begin{figure}[htb!]
\centering
\includegraphics[width=\textwidth]{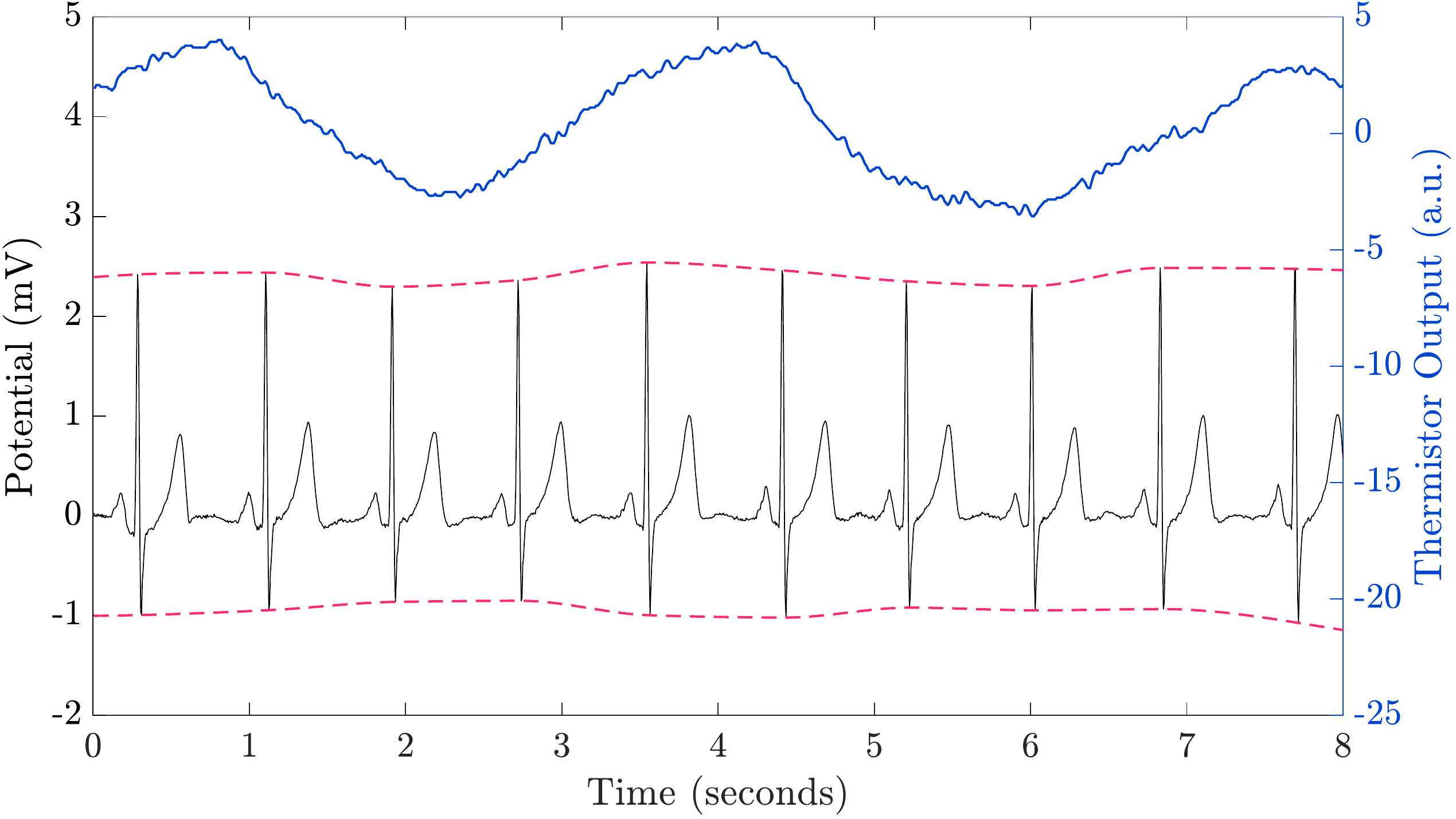}
\caption[The instantaneous amplitude of an ECG signal gives surrogate respiratory information]{The instantaneous amplitude of an ECG signal gives surrogate respiratory information. Top: we plot a respiratory flow signal {(nasal thermistor recording temperature differential)} in blue; bottom: we plot the simultaneously-recorded ECG signal and illustrate its amplitude modulation in a traditional way using the dotted pink lines.}
\label{edrillustration}
\end{figure}

\begin{example}
While the pulse rate of the sinoatrial (SA) node is constant, heart rate is generally not constant \cite{shaffer2014healthy}. This discrepancy is caused by neural and neuro-chemical influences on the pathway from the SA node to the cardiac muscles. Heart rate can be modeled as the instantaneous frequency of the ECG signal (see Figure~\ref{ihrillustration}); variations in heart rate (which correspond to frequency modulation of the ECG signal) are the main object of study in the field of heart rate variability (HRV) analysis \cite{shaffer2014healthy}.
\end{example}

\begin{figure}[htb!]
\centering
\includegraphics[width=\textwidth]{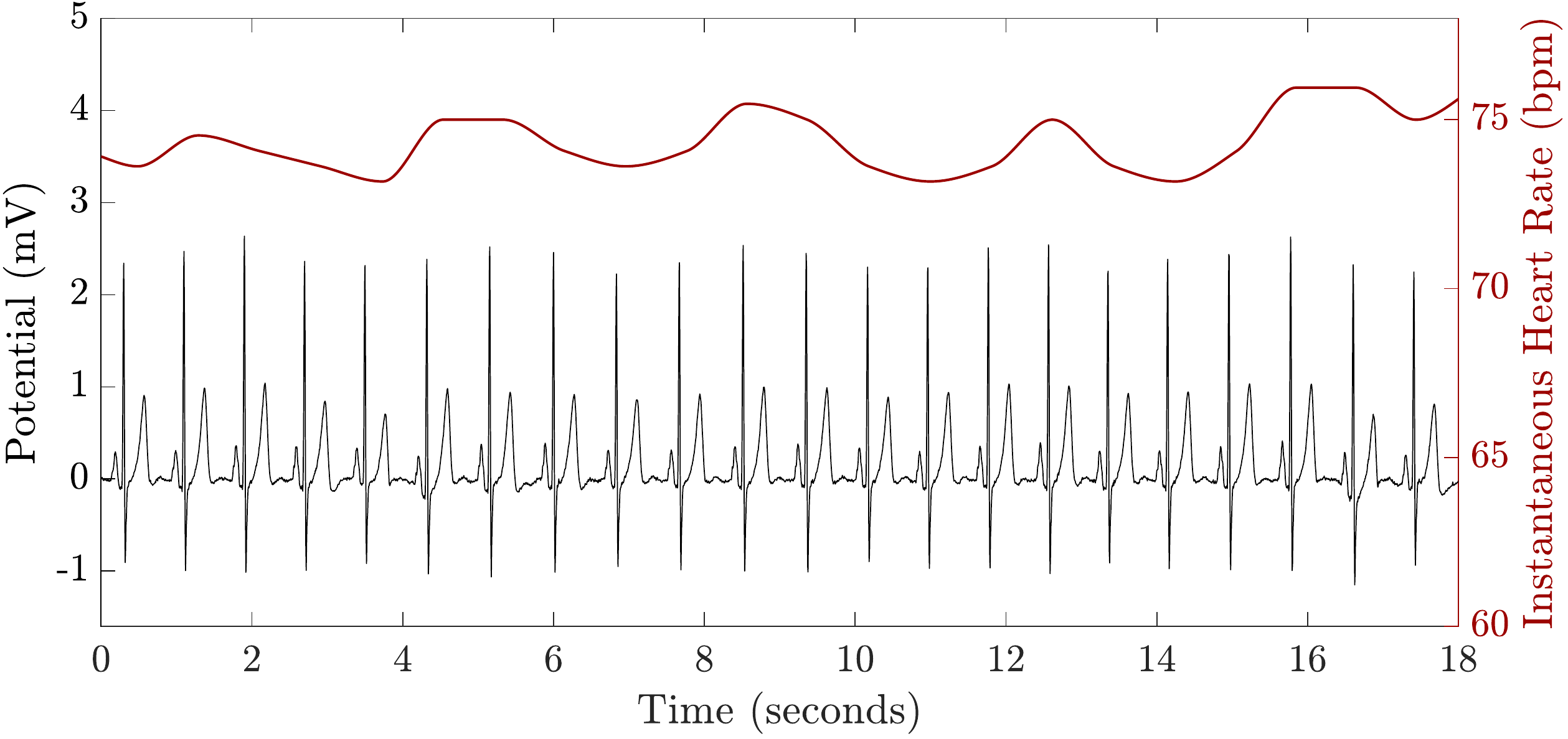}
\caption[The instantaneous frequency of an ECG signal is known as the instantaneous heart rate]{The instantaneous frequency of an ECG signal is known as the instantaneous heart rate. In red, we plot the instantaneous heart rate signal corresponding to the ECG signal in black.}
\label{ihrillustration}
\end{figure}

The wave-shape function $s$ deserves some more discussion.  In particular, we discuss the physiological information it encodes and how recovering $s$ is akin to traditional electrophysiological practices.  When analyzing ECG, electrocardiologists derive physiological information from the cardiac waveform independently of the amplitude and frequency modulation inherent in the signal. Their approach could be called a {\em landmark} approach to biomedical time series analysis, wherein they notice that the components of the cardiac wave-shape (as manifested in the ECG) correspond to different stages of heart contraction (and hence the heart substructures that are active during those stages).  Landmarks on the enigmatic template are conventionally labeled as shown in Figure~\ref{Fig:ECGTemplate}, and each deflection corresponds to a particular stage of heart contraction. 

\begin{figure}[htb!]
\centering
\includegraphics[width=.5\textwidth]{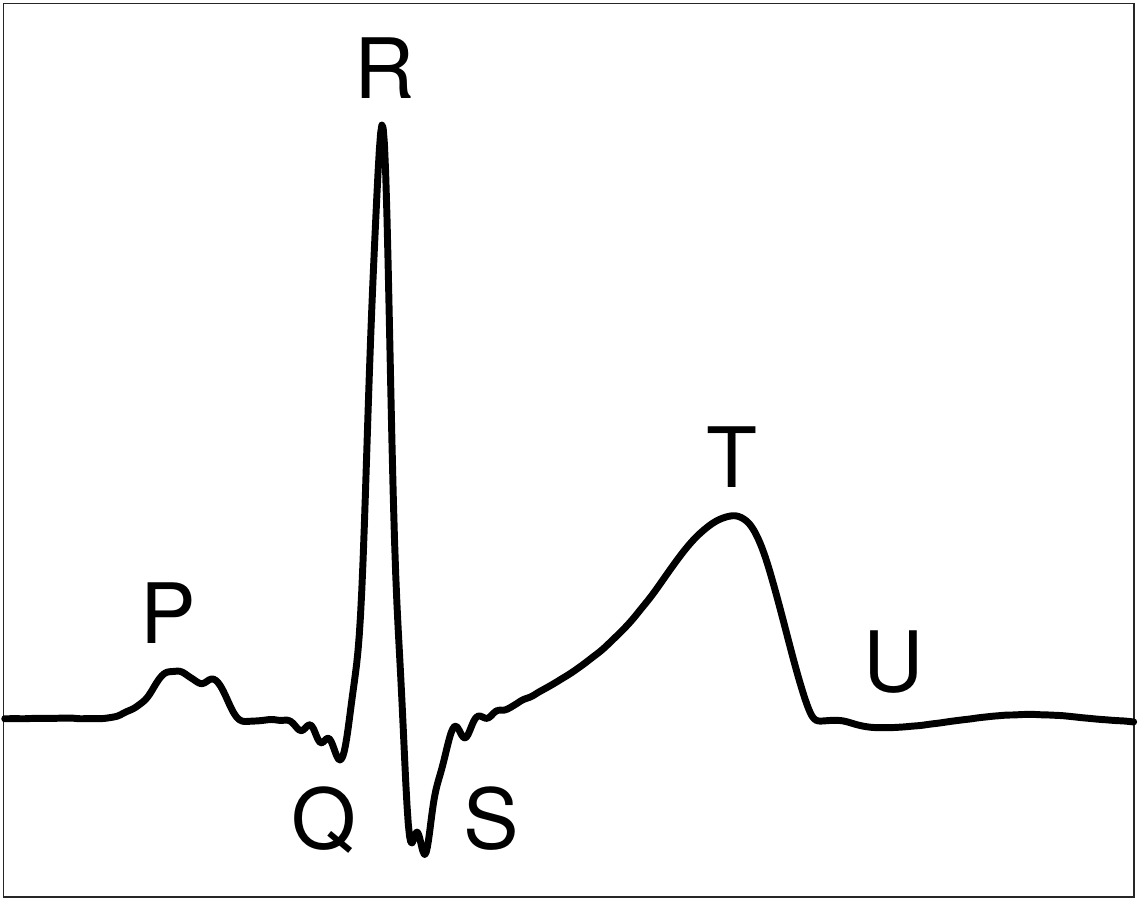}
\caption[Landmarks on the enigmatic ECG template correspond to distinct stages of a healthy heart contraction]{Landmarks on the enigmatic ECG template correspond to distinct stages of a healthy heart contraction.  The P wave corresponds to the depolarization of the atria, the QRS complex corresponds to the depolarization of the ventricles, the T wave corresponds to the repolarization of the ventricles, and the U wave corresponds to the repolarization of the papillary muscles.}
\label{Fig:ECGTemplate}
\end{figure}

\begin{example} When diagnosing atrial fibrillation (a cardiac arrhythmia associated with heart failure and stroke) using the ECG signal, cardiologists look for the absence of P waves, among other considerations (see Figure~\ref{Fig:AfibExample}).  The cardiac cycle is assessed independently of the frequency modulation caused by heart rate, the amplitude modulation caused by respiration, and the noise in the signal.  In atrial fibrillation, an additional confounding factor when recovering the cardiac wave-shape $s$ is the presence of a noise-like component called the fibrillatory wave ($f$-wave). Atrial fibrillation, including its modeling and analysis, is a challenging topic to which the wave-shape oscillatory model has been implicitly applied \cite{Malik_Reed_Wang_Wu:2017}. 
\end{example}

\begin{figure}[htb!]
\centering
\includegraphics[width=\textwidth]{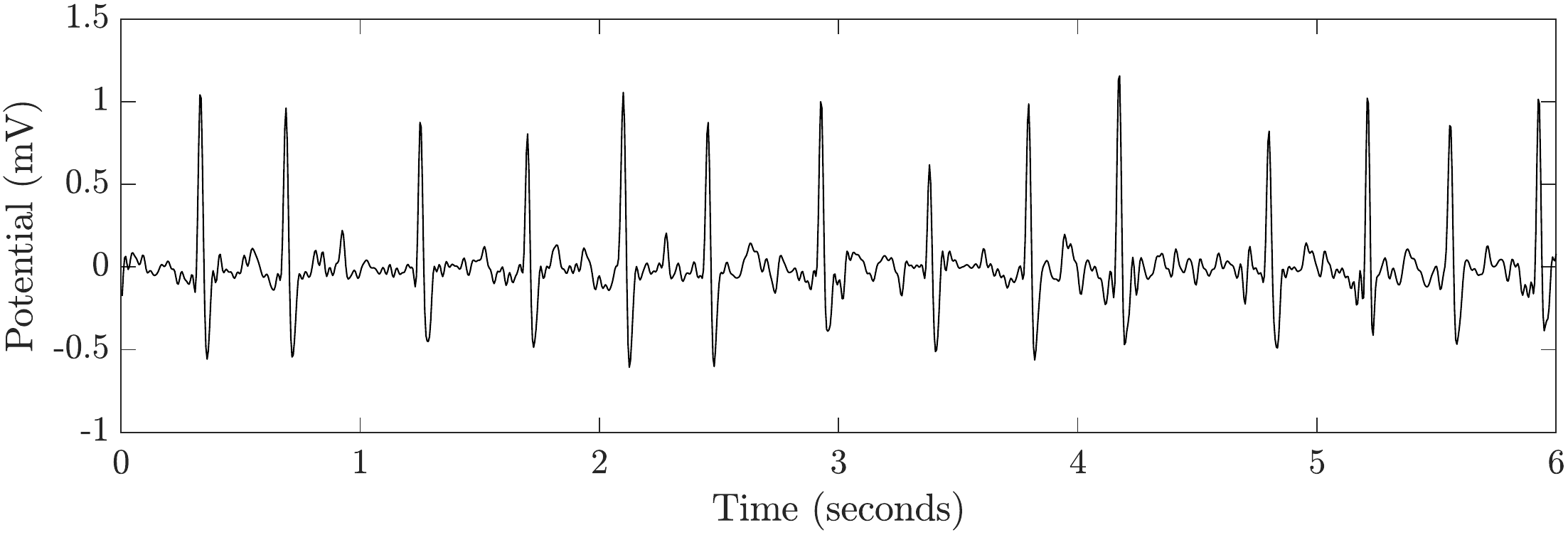}
\caption[Atrial fibrillation is characterized by a lack of P waves, irregularly-occurring ventricular contractions, and a fibrillatory wave]{Atrial fibrillation is characterized by a lack of P waves, irregularly-occurring ventricular contractions, and a rapid component called the fibrillatory wave.}
\label{Fig:AfibExample}
\end{figure}

\begin{example} When diagnosing myocardial infarction (blockage of the coronary artery), electrocardiologists look for a phenomenon called ST segment elevation, wherein the cardiac wave-shape in any precordial lead reaches an electric potential of approximately $0.2$ mV between the S and T landmarks.  While a high heart rate usually accompanies a heart attack, assessing the cardiac wave-shape independently of the heart rate allows physicians to detect the difference between, for example, high heart rates that accompany physical activity and high heart rates that accompany adverse cardiac events. %This phenomenon will be studied in Section~\ref{Sect:MI}. %The natural amplitude modulation that arises due to respiration is typically accounted for by measuring the average ST elevation over small collections of beats.
\end{example}

\begin{example} In ECG-based cardiac waveform analysis, the QT interval is the length of time between the start of the Q wave and the end of the T wave.  Since heart rate is known to be negatively correlated with cardiac cycle duration, QT interval lengths are commonly corrected so that the QT interval lengths of subjects with different heart rates can be effectively compared.  This correction process roughly amounts to recovering the duration of $s$ independently of any frequency modulation. 
An abnormal QT interval is associated with an increased risk for sudden cardiac death. In clinical trials for new medications, pharmaceutical companies use the corrected QT interval to assess the increased risk for sudden cardiac death that a patient taking the new drug will incur.
\end{example}

We return to a mathematical discussion of the wave-shape function. Due to the smoothness of $s$, it has a point-wise Fourier representation 
\begin{equation}\label{expansion of s}
s(t)=\alpha_0+\sum_{k=1}^\infty \alpha_k\cos(2\pi kt+\beta_k) \quad t\in \mathbb{R},
\end{equation}
where $\alpha_0\in \mathbb{R}$ and $\alpha_k\geq0$ are associated with the Fourier coefficients of $s$, and $\beta_k\in[0,2\pi)$. Hence, we have the following expansion for $a(t)s(\phi(t))$ in (\ref{Introduction:ANHM0}):
\begin{equation}\label{Introduction:ANHM1}
a(t)s(\phi(t))=a(t)\Big[\alpha_0+\sum_{k=1}^\infty \alpha_k\cos(2\pi k\phi(t)+\beta_k)\Big] \quad t \in \mathbb{R}.
\end{equation}
The signal $a(t)s(\phi(t))$ can be interpreted in two different ways. First, we could view it as an oscillatory signal with one oscillatory component; in this case, the oscillation is non-sinusoidal. Second, we could view it as an oscillatory signal with multiple oscillatory components, each having a cosine oscillatory pattern; in this case, we call the first oscillatory component $a(t)\alpha_1\cos(2\pi \phi(t)+\beta_1)$ the {\em fundamental component} and $a(t)\alpha_k\cos(2\pi k\phi(t)+\beta_k)$ (for $k \geq 2$) the {\em $k$-th multiple} of the fundamental component. Clearly, $\alpha_0$ is the zero-frequency term of the wave-shape function, and the instantaneous frequency of the $k$-th multiple is $k$-times that of the fundamental component. 
The second viewpoint is easier to analyze theoretically (e.g. via time-frequency methods such as the short-time Fourier transform and its associated synchrosqueezed representation). However, as described in the above examples, the first viewpoint is more physiological and coincides well with traditional electrophysiological practices. 
Physicians diagnose various cardiac diseases by reading the wave-shape function. However, physicians read the wave-shape in the time domain and not in terms of its Fourier decomposition. 
We have thus confirmed the utility of the wave-shape function for modeling and studying an oscillatory biomedical time series.

\subsection{Generalizing the Phenomenological Model}\label{section:generalization of the phenomenological model}
While the phenomenological model \eqref{Introduction:ANHM0} is able to capture the non-sinusoidal nature of cycles in an oscillatory biomedical time series, it does not capture time-dependent changes in oscillatory morphology (aside from those deviations due to amplitude and frequency modulation). 
Physiologically, this time-varying oscillatory morphology is important. Modeling changes in oscillatory morphology is especially relevant due to the increased prevalence of long-term, mobile cardiac monitoring. 

\begin{example} During a surgical procedure in which a general anesthetic is required, dosage is varied and administered based on the discretion of the supervising anesthesiologist or nurse anesthetist; vital signs are recorded for the duration of the procedure, which may be several hours.  The changing concentration of anesthetic in the body will continuously modulate the cardiac or pulse wave-shape of the patient. The dynamics underlying general anesthesia are discussed in \cite{Wang_Wu_Huang_Chang_Ting_Lin:2019}, and the analysis is extended to ABP signals.
\end{example}

\begin{example}
In the mobile monitoring of patients suspected of having heart disease, recordings are many hours long. However, adverse cardiac events such as paroxysmal supraventricular tachycardia or myocardial ischemia are only a few minutes in duration; the rest of the time, the patient appears normal. In this case, the cardiac wave-shape changes depending on the modulating cardiac health of the patient being monitored.
\end{example}

Motivated by the above examples, the phenomenological model is generalized to fully capture the time-varying morphology of cycles \cite{lin2018wave}. A similar model is considered in the followup research articles \cite{xu2018recursive,yang2017multiresolution}. The main idea is intuitive. 
We generalize $\{\alpha_k\}_{k=0}^\infty$ in \eqref{expansion of s} to be time-varying: for each $k \geq 0$ we let $a(t)\alpha_k$ be a new function $A_k(t)$.  We also let each $k\phi(t)+\beta_k/(2\pi)$ be a new function $\phi_k(t)$. We then have
\begin{equation}\label{Introduction:ANHM2}
a(t)s(\phi(t))=A_0(t)+\sum_{k=1}^\infty A_k(t)\cos(2\pi \phi_k(t)) \quad t \in \mathbb{R},
\end{equation}
where $A_k$ and $\phi_k$ satisfy a few conditions. In addition to the slowly varying assumption imposed in the phenomenological model, we have 
\begin{itemize}
\item $|\phi_k'(t)-k\phi_1'(t)|\leq \epsilon \phi_1'(t)$ for all $k\in \mathbb{N}_+$, where $\phi_1'>0$;
\item $A_k(t)\leq c_kA_1(t)$ for all $k\in \mathbb{N}_+$, where $A_1>0$ and $\{c_k\}_{k=1}^\infty$ is an $\ell^1$ sequence;
\item $|\phi_k''(t)|\leq \epsilon k\phi_1'(t)$ for all $k\in \mathbb{N}_+$;
\item $|A_k'(t)|\leq \epsilon c_k\phi_1'(t)$ for all $k \in \mathbb{N}_+$.
\end{itemize}
Other regularity conditions listed in \cite[Definition 2.1]{lin2018wave} also hold. Here, the conditions $|\phi_k'(t)-k\phi_1'(t)|\leq \epsilon \phi_1'(t)$ and $A_k(t)\leq c_kA_1(t)$ capture the fact that the wave-shape is not fixed. 
The conditions $|\phi_k''(t)|\leq \epsilon k\phi_1'(t)$ and $|A_k'(t)|\leq \epsilon c_k\phi_1'(t)$ capture the fact that the wave-shape may not change dramatically from one cycle to the next. 

While this model has been used to design algorithms which handle various physiological problems such as fetal ECG analysis \cite{SuWu2017}, fetal magnetocardiography \cite{escalona2018comparison}, simultaneous heart rate and respiratory rate estimation from the PPG \cite{cicone2017nonlinear}, and cardiogenic artifact recycling \cite{lu2019recycling}, its dependence on the slowly varying wave-shape assumption limits its applicability to physiological time series. Moreover, since the generalized phenomenological model \cite{lin2018wave} encodes the time-varying nature of the wave-shape function in the frequency domain, the wave-shape function interacts with the instantaneous frequency in a non-trivial way, which further limits its ability to independently quantify dynamics encoded by the instantaneous frequency and the time-varying wave-shape. 
Specifically, when the subject is not of normal physiology, the model is limited.

\begin{example} Premature ventricular contractions are heartbeats which are not triggered by the SA node but instead originate in the ventricles (the lower compartment of the heart). The frequency at which these beats occur is associated with congestive heart failure \cite{dukes2015ventricular}.  Premature ventricular contractions manifest as the heart skipping a beat and appear at seemingly random times. Morphologically, they appear very different from beats triggered by the SA node, and a signal featuring premature ventricular contractions will not adhere to the slowly varying wave-shape assumption (see Figure~\ref{PVCExample}).  A patient may have multiple morphologically distinct types of premature ventricular contractions, in which case his or her risk for cardiac disease increases. %Premature ventricular contractions will be discussed later in Section~\ref{Sect:PVC}.\\
\end{example}

\begin{figure}[htb!]
\centering
\includegraphics[width=\textwidth]{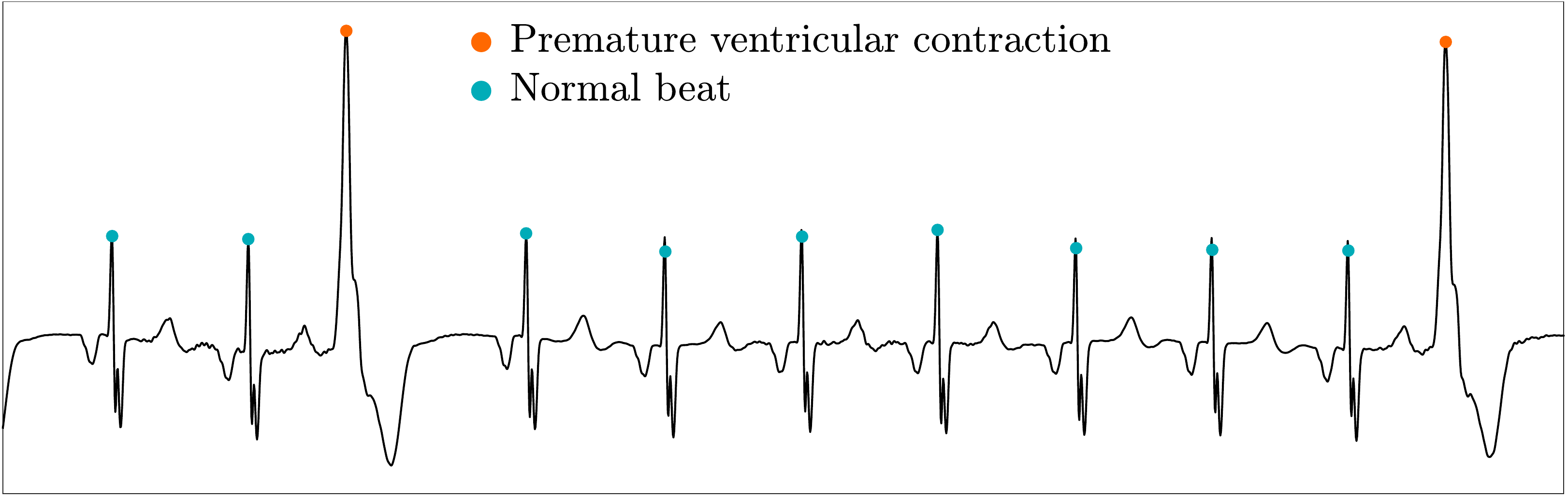}
\caption[ECG signals featuring premature ventricular contractions are difficult to model because the wave-shape is not slowly varying]{ECG signals featuring premature ventricular contractions are difficult to model because the wave-shape is not slowly varying. Premature ventricular contractions are morphologically distinct from normal beats and are not triggered by the SA node.}
\label{PVCExample}
\end{figure}

We refer readers with interest in the generalized phenomenological model to \cite{lin2018wave,xu2018recursive,yang2017multiresolution} for its theoretical details and to \cite{SuWu2017,escalona2018comparison,cicone2017nonlinear,lu2019recycling} for applications.
As useful as this model is, it cannot capture the (non-slowly) time-varying morphology of cycles in an oscillatory biomedical time series. Specifically, we have interest in those dynamics which are encoded by variations in cycle morphology that are independent of amplitude and frequency modulation.

\subsection{Connecting the phenomenological model to the wave-shape oscillatory model}\label{Section:Connecting2models}

%We now provide an argument showing that the wave-shape manifold is intimately related to the phenomenological model \eqref{Introduction:ANHM0}. 
%To this end, 
We show that for a given signal satisfying the phenomenological model, the collection of all oscillatory patterns can be well-approximated by a compact, two-dimensional, one-chart Riemannian $\mathcal{C}^1$-manifold.
The slowly varying assumption is used to show that each oscillatory pattern in the physiological signal is close to one whose amplitude and frequency are constant. The following theorem says that the collection of all such constant-amplitude, constant-frequency oscillatory patterns is a one-chart manifold. We use this opportunity to refer to the simulation of such a one-chart wave-shape manifold in Figure~\ref{Figure:waveshape manifold simulation}.

\begin{theorem}\label{theorem: simplified wave-shape manifold}
Suppose $s \in \mathcal{C}^2(\mathbb{R})$ is a function whose support is a subset of $[-\sfrac{1}{2},\sfrac{1}{2}]$. Let $U = I_1 \times I_2$, where $I_1\subset (0,\infty)$ and $I_2 \subset (1,\infty)$ are two open intervals of finite lengths. Define a map $\Phi \colon  U \rightarrow \mathcal H\subset L^2(\mathbb{R})$ by sending an ordered pair $(a, f) \in U$ to the function
$\mathbb{R} \ni t\mapsto a s ( f t ).$
Then $\Phi$ is a $\mathcal{C}^1$-diffeomorphism onto an open subset $\mathcal M:=\Phi(U)\subset L^2(\mathbb{R})$; that is, $\mathcal M$ is a manifold with one chart.
\end{theorem}

\begin{proof}
Note that since $s\in \mathcal{C}^{2}_c(\mathbb{R})$, $\Phi(U)\subset \mathcal{C}^{2}_c(\mathbb{R})$. 
Suppose $(a_1,f_1)\neq (a_2,f_2)$ while
\begin{equation}
\Vert \Phi(a_1, f_1) - \Phi(a_2, f_2) \Vert_2 = 0.
\end{equation}
We abuse the notation and denote $\Phi(a_1, f_1)=a_1s(f_1t)$ to simplify the discussion when there is no danger of confusion.
Since $s$ continuous, we have
\begin{equation} \label{proof:prop1:eq1}
a_1 s ( f_1 t ) = a_2s ( f_2 t )
\end{equation} 
for all $t \in \mathbb{R}$. 
Suppose $a_1 s ( f_1 t )$ has support $[a,b]\subset [-1/2f_1,1/2f_1]$. Without loss of generality, assume $f_2>f_1$, $a <0$, and $b>0$. The support of $a_2s(f_2t)$ is thus $[af_1/f_2,bf_1/f_2]\subset[a,b]$. Hence, we can find $x\in [a,af_1/f_2]\cup [bf_1/f_2,b]$ so that $a_1 s ( f_1 x) \neq 0$ but $a_2s ( f_2 x)=0$, which contradicts \eqref{proof:prop1:eq1}. As a result, we have shown that $\Phi$ is one-to-one and onto $\mathcal M:=\Phi(U)$.

We claim that the total differential of $\Phi$ is
\begin{equation}
D\Phi|_{(a_0, f_0)} (h_1, h_2) = h_1 s(f_0t) + h_2 a_0t s'(f_0t),
\end{equation}
where $(a_0,f_0)\in U$ and $(h_1,h_2)\in \mathbb{R}^2$. Note that $\Phi$ is a Hilbert space-valued function. Indeed, when $(h_1,h_2)=(a-a_0,f-f_0)$, we have 
\begin{align}\label{TotalDifferential}
&\Phi(a,f)-\Phi(a_0,f_0)-D\Phi|_{(a_0, f_0)} (a-a_0,f-f_0)\\
=&\,as(ft)-a_0s(f_0t)-[(a-a_0)s(f_0t)+(f-f_0)a_0ts'(f_0t)]\nonumber\\
=&\,a_0[s(ft)-s(f_0t)-(f-f_0)ts'(f_0t)]+(a-a_0)(s(ft)-s(f_0t))\nonumber.
\end{align}
By the integral form of Taylor's expansion, we have 
\begin{gather}
s(ft)-s(f_0t)=\int_{f_0}^fts'(zt)\,dz
\end{gather}
and 
\begin{gather}
s(ft)-s(f_0t)-(f-f_0)ts'(f_0t)=\int_{f_0}^f(f-z)t^2s''(zt)\, dz.
\end{gather}
Here, we view $s(ft)$ as a function of $f$. Next, we bound the $L^2$ norm of \eqref{TotalDifferential}. Indeed, we have
\begin{align}
&\Big\|\int_{f_0}^f(f-z)t^2s''(zt)\, dz\Big\|_{L^2}\\
=&\,\int\int_{f_0}^f\int_{f_0}^f(f-z)t^2s''(zt)(f-z')t^2s''(z't)\, dzdz'dt\nonumber\\
=&\,\int_{f_0}^f\int_{f_0}^f(f-z)(f-z')\Big[\int t^2s''(zt)t^2s''(z't)\, dt\Big]dzdz'\nonumber.
\end{align}
Note that since $s\in \mathcal{C}^2_c(\mathbb{R})$, we have
\begin{align}
\Big|\int t^2s''(zt)t^2s''(z't)\, dt\Big|&\leq \Big(\int t^4|s''(zt)|^2\, dt\Big)^{1/2}\Big(\int t^4|s''(z't)|^2\, dt\Big)^{1/2}=\frac{C}{zz'}\nonumber
\end{align}
for some $C>0$ depending on the fourth absolute moment of $s''$ only. As a result,
\begin{align}
\Big\|\int_{f_0}^f(f-z)t^2s''(zt)\, dz\Big\|_{L^2}&\leq C\int_{f_0}^f\int_{f_0}^f\frac{|(f-z)(f-z')|}{zz'}\, dzdz'\\
&\leq \frac{C}{f_0^2}\Big(\int_{f_0}^f|f-z|\, dz\Big)^2=\frac{C|f-f_0|^4}{4f_0^2}.\nonumber
\end{align}
Similarly, we can bound $\int_{f_0}^fts'(zt)\, dz$.
Therefore, 
\begin{gather}
\|\Phi(a,f)-\Phi(a_0,f_0)-D\Phi|_{(a_0, f_0)} (a-a_0,f-f_0)\|_{L^2}\leq C|f-f_0|^4,
\end{gather} which leads to 
\begin{gather}
\frac{\|\Phi(a,f)-\Phi(a_0,f_0)-D\Phi|_{(a_0, f_0)} (a-a_0,f-f_0)\|_{L^2}}{\|(a-a_0,f-f_0)\|}\to0
\end{gather} when $\|(a-a_0,f-f_0)\|\to 0$.
To finish the proof, we show that the total differential of $\Phi$ at $(a,f)\in U$, $D\Phi|_{(a, f)}$, is of full rank for any $(a,f)\in U$. 
It suffices to show that $s(ft)$ and $at s'(ft)$ are linearly independent in $L^2(\mathbb{R})$.  Suppose there are constants $c_1, c_2 \in \mathbb{R}$ such that for all $t \in \mathbb{R}$, 
\begin{equation}
c_1s(ft) = c_2ats'(ft).
\end{equation}
Suppose $c_1\neq 0$. In this case, $c_2\neq 0$; otherwise, we have a contradiction.
Since $s\in \mathcal{C}^2_c(\mathbb{R})$, there exists $t\neq 0$ such that $s(ft)\neq 0$ is the extremal value of $s$; that is, $fs'(ft)=0$. Since $f>0$, $s'(ft)=0$. Therefore, we have $c_1s(ft)\neq 0$ but $c_2ats'(ft)=0$, which is a contradiction. As a result, we conclude that $c_1=0$. In this case, $c_2$ must be $0$; otherwise, we have a contradiction. This concludes the claim that $D\Phi|_{(a, f)}$ is of full rank.
We conclude that $\Phi \colon U \rightarrow \mathcal{M}$ is a differomorphism.
\end{proof}

See Figure~\ref{Figure:waveshape manifold simulation} for an example of the wave-shape manifold in Theorem~\ref{theorem: simplified wave-shape manifold}. 
We can clearly see the nonlinear structure of the one-chart wave-shape manifold determined by the wave-shape function $s$ (the \texttt{db4} wavelet). We scale and dilate $s$ by independently sampling amplitudes uniformly from $I_1 = \left[ \sfrac{3}{4}, \sfrac{5}{4} \right]$ and frequencies uniformly from $I_2 = [1, 9]$. We visualize the wave-shape manifold by linearly projecting a set of $N = 2500$ discretized wave-shape functions from $\mathbb{R}^{7169}$ to $\mathbb{R}^3$. Each point is a wave-shape function $t \mapsto as(ft)$; red corresponds to high frequencies, and blue corresponds to low frequencies.
Next, we show that for a signal satisfying the phenomenological model \eqref{Introduction:ANHM0}, it can be well-approximated by the manifold indicated in Theorem \ref{theorem: simplified wave-shape manifold}.

\begin{theorem}\label{Theorem2 waveshape approximation}
Take $\epsilon> 0$ to be sufficiently small. Consider $f \colon \Bbb R \rightarrow \Bbb R$ satisfying the phenomenological model \eqref{Introduction:ANHM0}:
\begin{equation}
f(t) = A(t)\, s\left(\phi(t)\right).
\end{equation}
Assume without loss of generality that $\inf_{t\in \mathbb{R}}A(t)>0$ and $\inf_{t\in \mathbb{R}}\phi'(t)= 1$.
Define $t_n:=\phi^{-1}(n)$, where $n\in \mathbb{Z}$, and define functions on $\mathbb{R}$ as
\begin{equation}
g_n(t):=\left\{
\begin{array}{cl}
A(t_n)s(\phi'(t_n)t) & \text{if } \frac{-1}{2\phi'(t_n)}\leq t\leq \frac{1}{2\phi'(t_n)}\\
0 & \text{otherwise,}
\end{array}\right.
\end{equation}   
and
\begin{equation}
f_n(t):=\left\{
\begin{array}{cl}
A(t_n+t)s(\phi(t_n+t)) & \text{if } \frac{-1}{2\phi'(t_n)}\leq t\leq \frac{1}{2\phi'(t_n)}\\
0 & \text{otherwise.}
\end{array}\right.
\end{equation}   
We then have uniformly over $n$ that
\begin{equation}
\|g_n-f_n\|_\infty\leq C\epsilon\,,
\end{equation}
where $C=C(\|A\|_\infty,\|\phi'\|_\infty,\|s\|_{\mathcal{C}^1},\|\phi''\|_\infty)$. 
\end{theorem}

\begin{proof}
Note that for each $n\in \mathbb{Z}$, $g_n$ and $f_n$ are all supported in $[-\sfrac{1}{2},\sfrac{1}{2}]$. By a direct calculation, we have
\begin{align}
|g_n(t)-f_n(t)| \leq &|A(t_n)-A(t_n+t)||s(\phi'(t_n)t)|\\
&+A(t_n+t)|s(\phi'(t_n)t)-s(\phi(t_n+t))|\nonumber\,.
\end{align}
for all $t\in [-\sfrac{1}{2},\sfrac{1}{2}]$, and $|g_n(t)-f_n(t)|=0$ for all $t\notin [-\sfrac{1}{2},\sfrac{1}{2}]$.
By a direct bound, we have
\begin{align}
|A(t_n)-A(t_n+t)|&\leq \int_{t_n}^{t_n+t}|A'(x)|\, dx\leq \epsilon \int_{t_n}^{t_n+t}[\phi'(t_n)+\int_{t_n}^x\phi''(z)\, dx]\, dx\nonumber\\
&\leq \epsilon\phi'(t_n)|t|+\epsilon \frac{1}{2}\|\phi''\|_\infty |t|^2\leq \frac{1}{2}\epsilon\Big[\phi'(t_n)+\frac{1}{4}\|\phi''\|_\infty \Big]\,,
\end{align}
where the last bound holds since we only need to control $t\in[-\sfrac{1}{2},\sfrac{1}{2}]$. For the other term, note that 
\begin{align}
s(\phi(t_n+t))&=s\left( \phi(t_n)+\phi'(t_n)t+\int^{t_n+t}_{t_n}\phi''(z)z\, dz \right)\\
&=s\left(\phi'(t_n)t+\int^{t_n+t}_{t_n}\phi''(z)z\, dz\right)\nonumber
\end{align}
since $\phi(t_n)\in \mathbb{Z}$ and $s$ is $1$-periodic. 
Therefore, by denoting $J:=\int^{t_n+t}_{t_n}\phi''(z)z\, dz$, we have
\begin{align}
|s(\phi'(t_n)t)-s(\phi(t_n+t))|&\leq \int^{t_n+J}_{t_n}|s'(x)|\, dx=J\|s'\|_\infty.
\end{align}
$J$ is controlled in the same way:
\begin{align}
|J|&\leq \epsilon \int^{t_n+t}_{t_n}\phi'(z)z\, dz\leq \epsilon \int^{t_n+t}_{t_n}z\Big[\phi'(t_n)+\int^z_{t_n}\phi''(w)\, dw\Big] dz\nonumber\\
&\leq \epsilon\Big(\phi'(t_n)|t|^2+\frac{1}{3}\|\phi''\|_\infty|t|^3\Big)\leq  \frac{1}{4}\epsilon\Big(\phi'(t_n)+\frac{1}{6}\|\phi''\|_\infty\Big).\nonumber
\end{align}
As a result, we obtain the claim with $C=C(\|A\|_\infty,\|\phi'\|_\infty,\|s\|_{\mathcal{C}^1},\|\phi''\|_\infty)$.
\end{proof}

A direct consequence of this theorem is that $\{g_n\}_{n\in \mathbb{Z}}$ is a subset of the one-chart manifold $\mathcal M:=\Phi(U)$ described in Theorem \ref{theorem: simplified wave-shape manifold}, where $U=I_a\times I_f$, $I_a=(\inf_n A(t_n),\sup_n A(t_n))$ and $I_f=(1,\sup_n \phi'(t_n))$. As a result, the collection of oscillatory cycles, $\{f_n\}_{n\in \mathbb{Z}}$, can be parametrized by $\mathcal M$ up to a controllable error depending on $\epsilon C(\|A\|_\infty,\|\phi'\|_\infty,\|s\|_{\mathcal{C}^1},\|\phi''\|_\infty)$. We mention that a similar argument can be applied to the signal satisfying the generalized phenomenological model summarized in Section \ref{section:generalization of the phenomenological model} with more tedious notation and calculation. Since it does not shed more light upon the topic, we omit the details.

\section{Theoretical support for the proposed DDmap algorithm}\label{Section:Theory}
{The {DDmap} algorithm is proposed as a tool for recovering the wave-shape manifold and hence the dynamics along its surface. Behind the {DM} algorithm are a number of theoretical results which guarantee the effectiveness {of DDmap} and which we summarize here.}
We need the following two assumptions regarding the wave-shape manifold.

\begin{assumption}\label{Assumption:manifold}
Assume the wave-shape manifold $\mathcal M_f$ is an $m$-dimensional, closed (compact without boundary){,} and smooth Riemmanian manifold embedded in $\mathbb{R}^p$ with the Riemmanian metric $g$ induced from the canonical metric of $\mathbb{R}^p$, where $m\leq p$. 
We assume that $\mathcal{S}=\{s_i\}_{i=1}^n$ is independently and identically sampled from a random vector $S:(\Omega,\mathcal{F},\mathbb{P})\to \mathbb{R}^p$, where the range of $S$ is supported on {$\mathcal{M}_f$}.
%, and the noise $\xi_i$ is independent of $S$. 
\end{assumption}

\begin{assumption}\label{Assumption:pdf}
We assume that the induced measure {$S_*\mathbb{P}$} on the Borel sigma algebra on $\mathcal M_f$ is absolutely continuous with respect to the Riemannian measure $dV_g$. Furthermore, we assume that the function $\mathsf{p}:=\frac{S_*\mathbb{P}}{dV_g}: \mathcal{M}_f \to \mathbb{R}^+$ given by Radon-Nikodym theorem is bounded away from zero and is sufficiently smooth. We call $\mathsf{p}$ the {\em probability density function} (p.d.f.) on $\mathcal M_f$ associated with ${S}$. When $\mathsf{p}$ is a constant function, we say ${S}$ is uniform; otherwise{,} ${S}$ is nonuniform.
\end{assumption}

\subsection{Recovery of the dynamics}\label{Sec:theoreticalSupportof DM}

Denote by $g_k$ and $\mu_k$ the $k$-th eigenfunction and eigenvalue of the Laplace-Beltrami operator $\Delta_g$ of the $d$-dim wave-shape manifold $\mathcal{M}_f$, where $k\in \mathbb{N}$; that is $\Delta_g g_k=-\mu_kg_k$. According to basic elliptic theory \cite{Berard:1986}, the spectrum of $\Delta_g$ is discrete and accumulates at $\infty$; that is, $\mu_1=\mu_2=\ldots=\mu_{n_c}=0<\mu_{n_c+1}\leq \mu_{n_c+2},\ldots$, where $n_c\in \mathbb{N}$ is the number of connected components of $\mathcal{M}_f$, and the dimension of each eigenspace, denoted as $E_k$, is finite, except at the accumulation eigenvalue.  
We denote by $\phi_k$ and $\tilde\lambda_k$ the $k$-th eigenvector and eigenvalue of the GL $h^{-1}(I - D^{-1}W)$ constructed from the point cloud $\mathcal{S}_f=\{s_i\}_{i=1}^n\subset \mathbb{R}^p$ \eqref{Definition:Atransition}. Note that $\lambda_k = 1 - h\tilde\lambda_k$ and that the GL and the diffusion operator $A$ have the same eigenvectors.

Now, assume Assumptions \ref{Assumption:manifold} and \ref{Assumption:pdf} hold, and assume that $\mathsf{p}$ is uniformly bounded from below and above. Suppose $\varphi:\mathbb{R}\to \mathcal{M}_f$ represents the dynamics in which we have interest so that $\varphi(t_j) = s_j$; that is, $s_j$ is located at time $t_j$.
We now claim that we can recover the dynamics $\varphi$ by the DM. 

Under the above setup, recall the recent $L^\infty$ spectral convergence result reported in \cite[Theorem 2]{dunson2019diffusion}. Fix $K\in \mathbb{N}$ and consider $h=h(n)$ so that $h=n^{-\frac{1}{4{m}+15}}$. When $n$ is sufficiently large, depending on $K$, 
%
%For $\mathcal{M}_f$, denote by $V$ its volume, $\iota_0$ its injectivity radius, $K$ {a} global upper bound on {its} sectional curvature, and $R$ its reach. Take $\beta>1$.  
For all $k\leq K$, with probability at least $1-n^{-2}$, we have
\begin{equation}\label{embedding spectral convergence eigenvalue}
|\tilde\lambda_k-\mu_k|\leq \Omega_1\epsilon^{3/2}
\end{equation} 
and
\begin{align}\label{embedding spectral convergence eigenvector}
\max_{s_i}|a_{k}\phi_k(i)-g_k(x_i)| \leq \Omega_2 \epsilon^{1/2}\,,\nonumber
\end{align}
where $\Omega_1$ and $\Omega_2$ depend on the curvature and density function. We refer readers with interest to \cite[Theorem 2]{dunson2019diffusion} for a detailed description of all relevant quantities.
As a result, for each fixed $k\leq K$, we have
\begin{equation}
|\tilde\lambda_k-\mu_k|\to 0\mbox{ and }\max_{s_i}|\phi_k(i)-g_k(s_i)|\to 0
\end{equation} 
when $n\to \infty$ almost surely. 
This spectral convergence result of the GL emphasizes what $\phi_k$ estimates.
Particularly, by composing the eigenvectors and the temporal information, the above shows that $\phi_k=[\phi_k(1),\ldots,\phi_k(n)]^\top$ is actually an estimate of $[g_k\circ \varphi(t_1),\ldots,g_k\circ \varphi(t_n)]^\top$ but {\em not} a direct estimate of the intrinsic dynamics. However, it does contain useful information concerning the intrinsic dynamics; in this sense the product of the proposed algorithm is a surrogate for the dynamics of interest.

The above result can be immediately combined with the spectral embedding theory \cite{Berard:1986,Berard_Besson_Gallot:1994} to justify how the DM recovers the manifold, and hence the dynamics on it.
The {\em spectral embedding} of $\mathcal M_f$ is defined as follows \cite{Berard_Besson_Gallot:1994}. Take an $L^2(\mathcal{M}_f)$ basis $a=\{g_k\}_{k=1}^\infty\in \Pi_{k=1}^\infty O(\text{dim}(E_k))$. Define
\begin{equation}
\Psi^a_t:x\mapsto (2t)^{\frac{m+2}{4}}\sqrt{2}(4\pi)^{\frac{m}{4}}(e^{-t\mu_k}g_k(x))_{k=1}^\infty\in \ell_2,
\end{equation}
where $t>0$ is the {\em diffusion time}. 
It is shown in \cite[Theorem 5]{Berard_Besson_Gallot:1994} that $\Psi^a_t$ is not only an embedding for any $t>0$, but also an almost-isometric embedding when $t>0$ is sufficiently small; that is, the pulled-back metric $(\Psi^a_t)^*\texttt{can}$, where $\texttt{can}$ is the canonical metric on $\ell^2$, satisfies
\begin{equation}
(\Psi^a_t)^*\texttt{can}=g+\frac{2t}{3}\Big(\frac{1}{2}\text{Scal}_gg-\text{Ric}_g\Big)+O(t^2)
\end{equation} 
when $t\to 0$, where $\text{Scal}_g$ is the scalar curvature and $\text{Ric}_g$ is the Ricci curvature.
Recently, the spectral embedding theory of the Laplace-Beltrami operator was generalized to the finite-dimensional setting \cite{Bates:2014,Portegies:2015}. Denote by $\kappa$ a lower bound on the Ricci curvature of $\mathcal M_f$. Then, for a given tolerable error $\epsilon>0$, there exists a $t_0=t_0(m,\epsilon, \kappa,\iota_0)$ such that for all $0<t<t_0$, there exists an $N_E=N_E(m, \epsilon,t,\kappa,\iota_0,V)$, where $\iota_0$ is the lower bound of injectivity radius and $V$ is the upper bound of the volume, such that if $q\geq N_E$, the finite-dimensional map
\begin{equation}\label{Embededing finite spectral}
\Psi^{a,(q)}_t:x\mapsto \sqrt{2}(4\pi)^{\frac{m}{4}}(2t)^{\frac{m+2}{4}}(e^{-t\mu_k}g_k(x))_{k=1}^q\in \mathbb{R}^q
\end{equation}
is an embedding and satisfies
\begin{equation}\label{Embededing finite spectral almost isometric}
1-\epsilon\leq \|(d\Psi_t^{a,(q)})_x\|\leq 1+\epsilon.
\end{equation}
Note that the DM defined in (\ref{DM}) {truncated to $q$ coordinates} can be viewed as a discretization of ${\Psi^{a,(q)}_t}$ without the universal constant $(2t)^{\frac{m+2}{4}}\sqrt{2}(4\pi)^{\frac{m}{4}}$, while the eigenvalues and eigenvectors are estimated from the data. 

Now, we put everything together. For a given sufficiently small $\epsilon$, take $0<t<t_0(m,\epsilon, \kappa,\iota_0)$. Fix $q\geq N_E(m, \epsilon,t,\kappa,\iota_0,V)$. When $n$ is sufficiently large, with high probability, the spectral convergence \eqref{embedding spectral convergence eigenvalue} holds for all $k$-th eigenvalues and eigenvectors when $k\leq q$.
As a result, when $n$ is finite, with high probability, the DM recovers the manifold $\mathcal M_f$ due to the finite spectral embedding \eqref{Embededing finite spectral}.

\subsection{Robustness to noise}\label{Sec:theoreticalSupportof DM robustness}

A real physiological time series is inevitably noisy. 
The wave-shape $x_i\in \mathcal{X}_f$ is a noisy version of $s_i\in\mathcal{S}_f \subset\mathcal  M_f$.
When the data is noisy, it has been shown in \cite{ElKaroui:2010a,ElKaroui_Wu:2016b} that under some mild assumptions, the DM is robust to noise, which we summarize here. Denote by $W$ and $W_0$ the affinity matrices associated with $\mathcal{X}_f$ and $\mathcal{S}_f$ respectively. By \cite{ElKaroui:2010a,ElKaroui_Wu:2016b} when the connection group is the trivial $SO(1):=\{1\}$ group, if $\sup_{i,j}|W_{i,j}-(W_0)_{i,j}|\leq \epsilon$ for $\epsilon>0$, and $\inf_{i}\sum_{j\neq i}W_{i,j}/n>\gamma$ and $\gamma>\epsilon$,
we have 
\begin{equation}
\|A(W)-A(W_0)\|\leq \frac{\epsilon}{\gamma}\left(1+\frac{1}{\gamma-\epsilon}\right),
\end{equation}
where $\|\cdot\|$ is the operator norm.
By Weyl's inequality and the Davis-Kahan $\sin(\theta)$ theorem, the first $q'$ eigenvectors and eigenvalues are well-reconstructed up to a controllable error, where the number $q'$ depends on the noise level.
When $q'$ is large enough so that $q'\geq N_E(m, \epsilon,t,\kappa,\iota,V)$, with the finite-dimensional embedding result, we are guaranteed a reconstruction of the manifold. We thus conclude that the DM affords us a reconstruction of the clean manifold up to a tolerable error. 

\end{document}